\documentclass[bj]{imsart}

\newcommand\R{{\mathbb{R}}}

\newcommand\N{\mathbb{N}}

\newcommand\eps{\varepsilon}

\newcommand{\wt}{\widetilde}
\newcommand{\widebar}[1]{\mbox{\kern1.5pt\hbox{\vbox{\hrule height 0.6pt \kern0.35ex
        \hbox{\kern-0.15em \ensuremath{#1 }\kern0.0em}}}}\kern-0.1pt}

\newcommand{\E}{\mathbb{E}}

\newcommand{\dint}{{\rm d}}

\newcommand{\abs}[1]{\left\vert #1 \right\vert}
\newcommand{\norm}[1]{\left\Vert #1 \right\Vert}

\DeclareFontEncoding{FMS}{}{}
\DeclareFontSubstitution{FMS}{futm}{m}{n}
\DeclareFontEncoding{FMX}{}{}
\DeclareFontSubstitution{FMX}{futm}{m}{n}
\DeclareSymbolFont{fouriersymbols}{FMS}{futm}{m}{n}
\DeclareSymbolFont{fourierlargesymbols}{FMX}{futm}{m}{n}
\DeclareMathDelimiter{\VERT}{\mathord}{fouriersymbols}{152}{fourierlargesymbols}{147}

\RequirePackage[OT1]{fontenc}
\RequirePackage{amsthm,amsmath,amssymb,booktabs}
\RequirePackage[colorlinks,citecolor=blue,urlcolor=blue]{hyperref}

\arxiv{arXiv:1503.04123}

\startlocaldefs
\numberwithin{equation}{section}
\theoremstyle{plain}
\newtheorem{thm}{Theorem}[section]
\newtheorem{lem}{Lemma}[section]
\newtheorem{cor}{Corollary}[section]
\newtheorem{prop}{Proposition}[section]
\newtheorem{ass}{Assumption}[section]
\newtheorem{alg}{Algorithm}[section]
\newtheorem{ex}{Example}[section]

\theoremstyle{remark}
\newtheorem{rem}{Remark}[section]

\endlocaldefs

\begin{document}

\begin{frontmatter}
\title{Perturbation theory for Markov chains via Wasserstein distance}
\runtitle{Perturbation theory for Markov chains via Wasserstein distance}

\begin{aug}
\author{\fnms{Daniel} \snm{Rudolf}\thanksref{a}\ead[label=e1]{daniel.rudolf@uni-goettingen.de}}
\and
\author{\fnms{Nikolaus} \snm{Schweizer}\thanksref{b} \ead[label=e2]{n.f.f.schweizer@uvt.nl}}

\address[a]{Institut f\"ur Mathematische Stochastik, Universit\"at G\"ottingen, 
Goldschmidtstra\ss e 7, 37077
G\"ottingen, Germany.
\printead{e1}}

\address[b]{Department of Econometrics \& OR, Tilburg University, 
P.O.box 90153, 5000 LE Tilburg, The Netherlands.
\printead{e2}}

\runauthor{D. Rudolf and N. Schweizer}

\affiliation{Universit\"at G\"ottingen and University of Tilburg}

\end{aug}

\begin{abstract}
Perturbation theory for Markov chains addresses the question of
how small differences in the transition probabilities
of Markov chains are reflected in differences between their distributions.
We prove powerful and flexible bounds on the distance of 
the $n$th step distributions
of two Markov chains when one of them satisfies
a Wasserstein ergodicity condition.
Our work is motivated by the recent interest in 
approximate
Markov chain Monte Carlo (MCMC) methods
in the analysis of 
big data sets. 
By using an approach based on Lyapunov functions, we 
provide estimates for geometrically ergodic Markov chains 
under
weak assumptions. 
In an autoregressive model, our bounds cannot be improved in general. 
We illustrate our theory by showing quantitative estimates
for approximate 
versions of two prominent MCMC algorithms, 
the Metropolis-Hastings
and stochastic Langevin algorithms.
\end{abstract}


\end{frontmatter}

\section{Introduction}

Markov chain Monte Carlo (MCMC) algorithms 
are one of the key tools in computational statistics.
They are used for the approximation of expectations 
with respect to probability measures given by unnormalized
densities. For almost all classical MCMC methods
it is essential to evaluate
the target density. In many cases, 
this requirement is not an issue, but there are also 
important applications where it is a problem. 
This includes applications where the density is not available in closed form, see 
\cite{MaPuRoRy12},
or where an exact evaluation is computationally too demanding, see \cite{AFEB14}.
Problems of this kind lead to the approximation of Markov chains and
to the question of
how small differences in the transitions
of two Markov chains affect the differences between their distributions.

In Bayesian inference when \emph{big data} sets are 
involved an exact evaluation of the target density
is typically  
very expensive. 
For instance, in each step of a Metropolis-Hastings algorithm
the likelihood of a proposed 
state
must be computed. 
Every observation in the underlying data set contributes 
to the likelihood and must be taken into account in the calculation. 
This may result in evaluating
several terabytes of data in each step of the algorithm. 
These are the reasons for the recent   
interest in numerically cheaper approximations of 
classical MCMC methods, see 
\cite{BDH14,BaDoHo15,KCW14,SWM12,WeTe11}. 
A reduction of the computational costs can, e.g., be achieved 
by relying on a moderately sized random subsample of the 
data in each step of the algorithm. The function value of the 
target density is thus replaced by an approximation. Naturally, subsampling 
and alternative attempts at ``cutting the Metropolis-Hastings budget'' 
\cite{KCW14} induce additional biases. 
These biases can lead to dramatic changes in 
the properties of the algorithms as discussed in \cite{Be15}.

We thus need a better theoretical understanding of the behavior of such \emph{approximate MCMC} methods.
Indeed, a number of recent papers prove estimates of these biases,
see \cite{AFEB14,BDH14,JoMaMuDu15,LeDoLa14,MeLeRo15,PS14}. 
A key tool
in these papers are perturbation bounds for Markov chains. 
One such result for uniformly ergodic Markov chains due to Mitrophanov \cite{Mi05} 
is used  
in \cite{AFEB14}. 
A similar perturbation estimate implicitly appears in \cite{BDH14}.	  
The focus on uniformly ergodic Markov chains 
is rather restrictive, especially for high-dimensional, non-compact state spaces such as $\mathbb{R}^m$. 
Working
with Wasserstein distances has recently 
turned out to be a fruitful alternative in several contributions on high-dimensional 
MCMC algorithms, see \cite{DuMo15,Eb14,Gi04,HaStVo14,MaSe10}. 

We provide perturbation bounds based 
on Wasserstein distances, which 
lead to flexible quantitative estimates of the biases of approximate 
MCMC methods. 
Our first main result is the Wasserstein perturbation 
bound of Theorem \ref{thm: drift}. Under a Wasserstein ergodicity
assumption, explained in Section~\ref{secPrelim}, it provides 
an upper bound on the distance of the $n$th step distribution between 
an ideal and an approximating Markov chain in terms of
the difference between their one-step transition probabilities. 
The result is well-suited for applications 
on a non-compact state space, since the difference of the one-step transition probabilities 
is measured 
by a weighted supremum with respect to a suitable Lyapunov function.  
For an autoregressive model, 
we show in Section \ref{subsec: auto_regr_proc} 
that the resulting perturbation bound cannot be improved in general.
As a consequence of the Wasserstein approach we also obtain 
perturbation estimates for geometrically ergodic Markov chains. 
We first adapt our Wasserstein perturbation bound to this setting. 
Then, as a second main result, Theorem \ref{thm geom3}, 
we prove a refined 
estimate
for geometrically ergodic chains where the perturbation 
is measured by a weighted total variation distance.
Our perturbation bounds,
and earlier ones in \cite{Mi03, Mi05}, establish a direct 
connection between an exponential convergence property for Markov chains and their
robustness to perturbations. In particular, fast convergence
to stationarity implies insensitivity to perturbations in the transition probabilities.
Geometric ergodicity has been studied extensively in the 
MCMC literature. Thus, our estimates can be used in combination with many existing convergence 
results for MCMC algorithms. In Section \ref{sec: appl}, 
we illustrate the applicability of both theorems
by generalizing recent 
findings on 
approximate Metropolis-Hastings algorithms from \cite{BDH14} 
and on noisy Langevin algorithms
for Gibbs random fields
from \cite{AFEB14}.

\subsection{Related literature}

We refer to \cite{Ka86, KaGo13} for an overview of the 
classical literature on perturbation theory for Markov chains.
However, as Stuart and Shardlow observed in \cite{ShSt00}, the classical  assumptions on the 
perturbation  
might be too restrictive for many interesting applications. 
As a consequence,
 they develop a perturbation theory for geometrically ergodic Markov chains \cite{ShSt00} 
which requires to control perturbations  
of iterated transition kernels in a weaker sense.
In our bounds for geometrically ergodic Markov chains,
we have similar flexibility in the perturbation due to the Lyapunov-type stability condition, 
and require only a control on the errors of one-step transition kernels.

Mitrophanov, in \cite{Mi05}, considers uniformly ergodic
Markov chains 
and provides 
the best estimates in those settings. 
In the geometrically ergodic case, there are further 
related results, see \cite{FHL13} 
and the references therein. 
Compared to \cite{FHL13}, our 
focus is on non-asymptotic 
estimates
with explicit constants,
while their main focus is on qualitative results such as 
inheritance of geometric ergodicity by the perturbation. 
Earlier related results on perturbations induced by floating-point 
roundoff errors are shown in \cite{BrRoRo01,RoRoSch98}.

Finally, let us point out that our paper 
is complementary to 
the work of 
Pillai and Smith \cite{PS14} 
who also present Wasserstein perturbation bounds for Markov chains.
When moving beyond the uniformly ergodic Markov chain case, 
an important challenge is to handle the issue that 
in many applications suprema of relevant quantities 
over the whole state space are infinite. 
The authors of \cite{PS14} guarantee finiteness 
of supremum norms by restricting attention 
to subsets of the state space. 
Their bounds thus involve exit probabilities from these subsets. 
Our approach circumvents these issues by relying on 
Lyapunov-type stability conditions for the approximate algorithm.

\section{Wasserstein ergodicity}\label{secPrelim}
Let $G$ be a Polish space 
and $\mathcal{B}(G)$ be the corresponding Borel $\sigma$-algebra.
Let $d$ be a metric, possibly different from the one which makes the space Polish, 
which is assumed to be lower semi-continuous with respect to the product topology of $G$.
Let $\mathcal{P}$ be the set of all Borel probability measures on $(G,\mathcal{B}(G))$.
Then, we define 
the Wasserstein distance of $\nu,\mu\in \mathcal{P}$ by
\[
 W(\nu,\mu) = \inf_{\xi \in M(\nu,\mu)} \int_G \int_G d(x,y)\, \dint \xi(x,y),
\]
where $M(\nu,\mu)$ is the set of all couplings of $\nu$ and $\mu$,
that is, all probability measures $\xi$ on $G\times G$ with marginals $\nu$ and $\mu$.
Indeed, on $\mathcal{P}$ the Wasserstein distance satisfies 
the properties of a metric but is not necessarily finite, see \cite[Chapter~6]{Vi09}.	
For a measurable function $f\colon G \to \R$ we define
\[
 \norm{f}_{\rm Lip} = \sup_{x,y \in G, x\neq y} \frac{\abs{f(x)-f(y)}}{d(x,y)},
\]
which leads to the well-known duality formula
\begin{equation}\label{KRduality}
 W(\nu,\mu) 
 = \sup_{\norm{f}_{\rm Lip} \leq 1} \abs{\int_G f(x) (\dint \nu(x)- \dint \mu(x))}.
\end{equation}
For details we refer to \cite[Chapter~1.2]{Vi03}. 
By $\delta_x$ we denote the probability measure concentrated at $x$. Hence
$W(\delta_x,\delta_y) = d(x,y)$ is finite for $x,y\in G$.

Let $P$ be a transition kernel on $(G,\mathcal{B}(G))$ 
which defines a linear operator $P \colon \mathcal{P}\to \mathcal{P}$ given by
\[
 \mu P(A) = \int_G P(x,A)\,\dint \mu(x),\quad \mu\in \mathcal{P},\,A\in \mathcal{B}(G).
\]
With this notation we have $\delta_x P(A) = P(x,A)$.
Further, for a measurable function $f\colon G \to \R$
and $\mu\in \mathcal{P}$ we have
\[
 \int_G f(x)\, \dint (\mu P)(x) = \int_G P f(x) \, \dint \mu(x),
\]
with $P f(x) = \int_G f(y) P(x,\dint y)$ whenever one of the integrals exist,
see for example \cite[Lemma~3.6]{Ru12}.
Now, by 
\[
  \tau(P) := \sup_{x,y\in G, x\neq y} \frac{W(\delta_x P,\delta_y P)}{d(x,y)}
\]
we define the \emph{generalized ergodicity coefficient} 
of 
transition kernel $P$. 
This coefficient can be understood as a 
generalized Dobrushin ergodicity coefficient, see \cite{Do56a,Do56b}. 
Dobrushin himself called $\tau(P)$ the Kantorovich norm of $P$, 
see \cite[formula (14.34)]{Do96}. Finally, $\tau(P)$ also provides a lower bound of the coarse Ricci curvature of $P$ introduced
in \cite{Ol09}. 

Two essential properties of the ergodicity coefficient are submultiplicativity
and contractivity, see \cite[Proposition~14.3 and Proposition~14.4]{Do96}.
\begin{prop} \label{prop: contr_subm}
 For two transition kernels $P$ and $\wt P$ on $(G,\mathcal{B}(G))$ and $\mu,\nu\in \mathcal{P}$, we have
 \begin{align*}
  \tau(P \wt P) & \leq \tau(P) \tau(\wt P) \qquad \text{(Submultiplicativity)}, \\ 
\text{ and }\quad    W(\nu P,\mu P)
  & \leq \tau(P)\, W(\nu,\mu) \qquad \text{(Contractivity)}.
 \end{align*}
\end{prop}
As an immediate consequence of this contractivity, we obtain the following corollary.
 \begin{cor}\label{cor: Wass_erg}  
  Let $P$ be a transition kernel with stationary distribution $\pi$, i.e. $\pi P=\pi$,
  and assume
  for some (and hence any) $x_0\in G$ it holds that $\int_G d(x_0,x)\, \dint \pi(x)<\infty$.
  Then
  \begin{equation}  \label{eq: Wass_erg}
  \sup_{x\in G} \frac{W(\delta_x P,\pi)}{W(\delta_x,\pi)}    \leq \tau(P).
  \end{equation}
 \end{cor}
  \begin{proof}
  Because of the assumption $\int_G d(x_0,x)\, \dint \pi(x)<\infty$ we have
  that $W(\delta_x,\pi)$ is finite for any $x\in G$.
  Thus, the assertion follows by Proposition~\ref{prop: contr_subm} 
  and stationarity of $\pi$.
\end{proof}
 \begin{rem}
 For some special cases one also has an estimate of the form \eqref{eq: Wass_erg}
 in the other direction. 
 To this end, consider the trivial metric 
 $d(x,y)=2 \cdot \mathbf{1}_{x\not = y}$ with 
 indicator function
 \[
  \mathbf{1}_{x\not = y} = \begin{cases}
                            1 & x\not = y\\
                            0 & x=y.
                           \end{cases}
 \]
 Further, 
 let
  \[
   \norm{q}_\text{tv}:= \sup_{\norm{f}_\infty \leq 1} \abs{\int_G f(y)\, \dint q(y)}  = 2\sup_{A\in \mathcal{B}(G)} \abs{q(A)}
  \]
   be the total variation norm of a signed measure $q$ on $G$. 
   In this setting $W(\mu,\nu) = \norm{\mu-\nu}_{\text{tv}}$.
  For $x,y\in G$ with $x\not= y$ we have 
  $\norm{\delta_x - \delta_y}_{\text{tv}}=d(x,y)=2$
  so that
    \begin{equation} \label{eq: tv_dobrushin_erg_coef}
       \tau_1(P) = \frac{1}{2} \sup_{x,y \in G, x\not =y} \norm{\delta_x P - \delta_y P}_{\text{tv}}.
  \end{equation}
   The ``$1$'' in the subscript of $\tau_1(P)$
   indicates that we use the trivial metric.
   By applying the triangle inequality of the total variation
   norm we obtain $\tau_1(P) \leq \sup_{x\in G} \norm{\delta_x P - \pi}_{\text{tv}}$.
   If additionally $\pi$ is atom-free, i.e.,  $\pi(\{y\})=0$ for all $y\in G$, 
   we have
   $
    \norm{\delta_y - \pi}_{\text{tv}} = 2.
   $
   Then, the previous consideration and \eqref{eq: Wass_erg}
   lead to
   \[
    \frac{1}{2} \sup_{x\in G} \norm{\delta_x P - \pi}_{\text{tv}} \leq \tau_1(P) \leq \sup_{x\in G} \norm{\delta_x P - \pi}_{\text{tv}}.
   \]
   For the moment, let us
   assume that $P$ is uniformly ergodic, that is, there exist numbers $\rho \in [0,1)$ and $C\in (0,\infty)$
   such that
   \[
     \sup_{x\in G} \norm{\delta_x P^n - \pi}_{\text{tv}} \leq C \rho^n, \quad n\in \mathbb{N}.
   \]
 An immediate consequence of the uniform ergodicity
 is that
 $\tau_1(P^n)\leq C \rho^n$.
 \end{rem}
Also note that if 
there is an $n_0\in \mathbb{N}$ for which $\tau(P^{n_0})<1$ 
 we have by 
 the submultiplicativity, see Proposition~\ref{prop: contr_subm}, that $\tau(P^{n})$ converges exponentially to zero.
 This motivates to impose
the following assumption which contains the idea 
 to measure convergence of $\delta_x P^n$ to $\pi$ in terms of $\tau(P^n)$.
 \begin{ass}[Wasserstein ergodicity]
 \label{ass: wass_contr}
For the transition kernel $P$
 there exist
 numbers $\rho \in [0,1)$ and $C\in (0,\infty)$ such that
 \begin{equation} \label{eq: ricci_curv}
  \tau(P^n)=  \sup_{x,y\in G, x\not = y} \frac{W(P^n(x,\cdot),P^n(y,\cdot))}{d(x,y)} \leq C \rho^n, \quad n\in \mathbb{N}.
 \end{equation}
  \end{ass}

For any probability measure $p_0\in \mathcal{P}$, a transition kernel $P$ 
with stationary distribution $\pi$ 
and $p_n = p_0 P^n$ 
we have under the Wasserstein ergodicity condition that 
\begin{equation*}
\label{eq: pibound}
 W(p_n,\pi) 
 \leq C \rho^n W(p_0,\pi). 
\end{equation*}

\section{Perturbation bounds}\label{secBounds}

By $\mathbb{N}_0 = \{0,1,2,\dots\}$ we denote 
the non-negative integers and assume that
all random variables are defined on a common probability 
space $(\Omega,\mathcal{F},\mathbb{P})$
mapping to a Polish space $G$ equipped with a lower semi-continuous metric $d$.
Let the sequence of random variables $(X_n)_{n\in \N_0}$ 
be a Markov chain with transition kernel $P$ and initial distribution $p_0$, i.e.,
we have almost surely
\[
 \mathbb{P}(X_n\in A\mid X_0,\dots,X_{n-1}) = \mathbb{P}(X_n\in A\mid X_{n-1}) = P(X_{n-1},A), \qquad n\in \mathbb{N}
\]
and $p_0 (A)= \mathbb{P}(X_0\in A)$ for any measurable set $A\subseteq G$.
Assume that $(\wt X_n)_{n\in \mathbb{N}_0}$ 
is another Markov chain with transition kernel $\wt P$ 
and initial distribution $\wt p_0$.
We denote by $p_n$ the distribution of $X_n$ and by $\wt p_n$ 
the distribution of $\wt X_n$. 
Throughout the paper, $(X_n)_{n\in \mathbb{N}}$ is considered to be the ideal, 
unperturbed Markov chain 
we would like to simulate while $(\wt X_n)_{n\in \mathbb{N}_0}$ 
is the perturbed Markov chain that we actually implement.

\subsection{Wasserstein perturbation bound}\label{sec: basic bounds}

Similar as in  
\cite[Theorem~3.1]{Mi05}, 
we show quantitative bounds on the difference of $p_n$ and $\wt p_n$, but use the
Wasserstein distance instead of total variation.
Besides Assumption~\ref{ass: wass_contr}, 
the bounds depend on the difference of 
the initial distributions and on a suitably 
weighted one-step difference between $P$ and $\wt{P}$.

\begin{thm}[Wasserstein perturbation bound]
\label{thm: drift}
  Let Assumption~\ref{ass: wass_contr} be satisfied
  with the numbers $C\in (0,\infty)$ 
  and $\rho\in [0,1)$, i.e., $\tau(P^n) \leq C \rho^n$. 
  Assume that there are numbers $\delta\in (0,1)$ 
 	and $L\in(0,\infty)$ and a 
  	measurable 
 	Lyapunov function $\wt V:G \rightarrow [1,\infty)$ of $\wt P$ such that
 \begin{equation} \label{eq: drift_cond}
    (\wt P \wt V)(x) \leq \delta \wt V(x)+ L.
  \end{equation}
  Let 
  \[
   \gamma = \sup_{x\in G} \frac{W(\delta_x P,\delta_x \widetilde P)}{\wt V(x)}
  \qquad  \mbox{and} \qquad 
   \kappa=\max\left\{\wt p_0 (\wt V), \frac{L}{1-\delta} \right\}
  \]
  with ${\wt p_0}(\wt V) = \int_G \wt V(x)\, \dint {\wt p_0}(x)$.
  Then
  \begin{equation} \label{eq: wass_est_d}
    W(p_n,\widetilde p_n) \leq C\left( \rho^n W(p_0,\widetilde p_0) + (1- \rho^n)
     \frac{\gamma \kappa}{1-\rho} \right).
\end{equation}
\end{thm}
\begin{proof}
  By induction one can show that
  \begin{equation} \label{eq: repr}
     \widetilde p_n - p_n 
     = (\widetilde p_0 - p_0 )P^n 
     + \sum_{i=0}^{n-1} \widetilde p_i (\widetilde P-P)P^{n-i-1},\quad n\in \mathbb{N}.
  \end{equation}
  We have
\begin{align*}
W(\wt p_i P, \wt p_i \wt P)
 \leq \int_G W(\delta_x P, \delta_x \wt P)\, \dint \wt p_i(x)
   & \leq \gamma
     \int_G  \wt V(x)\, \dint \wt p_i(x).
  \end{align*}
Moreover, for $i \geq 0$ we have
\begin{align*}
 \int_G \wt V(x) \,\dint \wt p_i(x) 
 = \int_G \wt P^i \wt V(x)\, \dint \wt p_0(x)
 \leq \delta^i \wt p_0(\wt V) + \frac{L (1-\delta^i)}{(1-\delta)} 
 \leq \max\left\{ \wt p_0(\wt V),\frac{L}{1-\delta}\right\}
 \end{align*}
so that we obtain $W(\wt p_i P, \wt p_i \wt P) \leq \gamma \kappa$.
  By this fact 
  we have
  \begin{equation} \label{eq: est_w}
  W(\wt p_i \wt P P^{n-i-1}, \wt p_i P  P^{n-i-1}) 
      \leq \gamma \kappa \cdot \tau(P^{n-i-1}).
  \end{equation}
   Then, by \eqref{eq: repr}, \eqref{eq: est_w} and 
  the triangle inequality of the Wasserstein distance we have
  \begin{align*}
   W(p_n, \wt p_n) 
    & \leq 
    W(p_0 P^n,\wt p_0 P^n)
     + \sum_{i=0}^{n-1} 
     W(\wt p_i\wt P P^{n-i-1},\wt p_i P P^{n-i-1})\\
   & \leq W(p_0,\wt p_0 ) \tau(P^n) + \gamma \kappa \sum_{i=0}^{n-1} \tau(P^i).
  \end{align*}
Finally, by \eqref{eq: ricci_curv}
we obtain
$
 \sum_{i=0}^{n-1} \tau(P^i) \leq \frac{C(1-\rho^n)}{1-\rho},
$
which allows us to complete the proof.
 \end{proof}

\begin{rem}\label{rem: triv_drift}
The parameter $\kappa$ is an upper bound on $\wt{p}_i(\wt V)$. 
It can be interpreted as a measure for the stability of the perturbed Markov chain.  
The parameter $\gamma$
quantifies 
with a weighted supremum norm the one-step difference between 
$P$ and $\wt P$. 
The use of the Lyapunov function 
 increases the flexibility of the resulting estimate, since
 larger values of $\wt V$ compensate larger 
values of the Wasserstein distance between the kernels.
Notice that the existence of a Lyapunov function satisfying 
\eqref{eq: drift_cond} is 
weaker than assuming $\wt V$-uniform ergodicity of $\wt P$ since 
it is not associated with a small set condition. In particular, 
the condition is satisfied for any $\wt P$ with the trivial choice 
$\wt V (x) =1$ for all $x\in G$, see Corollary \ref{thm: was_mith}. 
As we will see in Section \ref{sec: appl}, 
allowing for non-trivial choices of $\wt V$ considerably increases the applicability of our results.
\end{rem}

If $\wt P$ has a stationary distribution, say $\wt \pi \in \mathcal{P}$, as a consequence
of the previous theorem,
we obtain bounds on the difference between $\pi$ and $\wt \pi$.
\begin{cor}  \label{cor: Wass_pi_pi_tilde}
 Let the assumptions of Theorem~\ref{thm: was_mith} be satisfied. 
 Assume that $\wt P$ has a stationary distribution $\wt \pi\in \mathcal{P}$ and let
 $W(\pi,\wt \pi)$ be finite. Then
 \begin{equation} \label{eq: dist_stat_dist2}
      W(\pi,\wt \pi) \leq \frac{\gamma C}{1-\rho} \cdot \frac{L}{1-\delta}.
 \end{equation}
\end{cor}
\begin{proof}
 By Theorem~\ref{thm geom3} we obtain with
 $p_0=\pi$, $\wt p_0 = \wt \pi$, the stationarity of the distributions $\pi$, $\wt \pi$ 
 and by letting $n \to \infty$ 
 that
 \[
  W(\pi,\wt \pi) \leq \frac{C \gamma \kappa}{1-\rho}.
 \]
 By the Lyapunov condition and \cite[Proposition~4.24]{Ha06}, it holds that
 \[
  \wt \pi(\wt V) = \int_G \wt V(x) \dint \wt \pi(x) \leq \frac{L}{1-\delta}
 \]
 which leads to $\kappa \leq L/(1-\delta)$ and finishes the proof.
\end{proof}
\begin{rem}
 It may seem artificial to assume $W(\pi,\wt \pi) < \infty$ but this is needed for the limit argument in the proof.
 This condition is often satisfied a priori. For example, it holds if
 the metric is bounded, i.e., $\sup_{x,y\in G} d(x,y)$ is finite, or, more generally, if 
 the distributions $\pi$ and $\wt \pi$ possess a first moment in the sense that
 there exist $x_0,\wt x_0 \in G$ such that
 \[
  \int_G d(x_0,x)\, \dint \pi(x) <\infty, \qquad \int_G d(\wt x_0,x )\,\dint \wt \pi(x) <\infty. 
 \]
\end{rem}
As pointed out in Remark \ref{rem: triv_drift}, we do not need to impose condition \eqref{eq: drift_cond}
to obtain a non-trivial perturbation bound:
 \begin{cor}  \label{thm: was_mith}
  Assume that Assumption~\ref{ass: wass_contr}
  holds with the numbers $C\in (0,\infty)$ 
  and $\rho\in [0,1)$, i.e., $\tau(P^n) \leq C \rho^n$, and
  let 
  \[
   \gamma := \sup_{x\in G} W(\delta_x P,\delta_x \widetilde P).
  \]
  Then
  \begin{equation} \label{eq: wass_est}
    W(p_n,\widetilde p_n) \leq C\left( \rho^n W(p_0,\widetilde p_0) + (1- \rho^n)
     \frac{\gamma}{1-\rho} \right).
  \end{equation}
 \end{cor}
 \begin{proof}
  The statement follows by Theorem~\ref{thm: drift} with $\wt V(x)=1$
  and $L=1-\delta$.
 \end{proof}
 
\begin{rem}
  For the trivial metric $d(x,y)=2 \cdot \mathbf{1}_{x\not = y}$ 
  the last corollary states essentially the result of \cite[Theorem~3.1]{Mi05},
  where instead of the general Wasserstein distance 
  the total variation distance is used. There, the bound's 
  dependence on $C$ and $\rho$ can be further improved 
  by using the a priori bound $\tau_1(P^n)\leq 1$ in addition 
  to uniform ergodicity. 
    For another metric $d$ such an a priori bound is in general not available. 
\end{rem}

\begin{rem}
Table~\ref{table: comp} provides 
a detailed comparison between our 
Theorem~\ref{thm: drift} and the related 
Wasserstein perturbation result of Pillai and 
Smith, \cite[Lemma~3.3]{PS14}. An important 
ingredient in their estimate is a set $\widehat G\subseteq G$
which can be interpreted as the part of $G$ where 
both Markov chains remain with high probability. 
When a good uniform upper bound on  
$W(\delta_x P,\delta_x \wt P)$ for all $x \in G$ is available, 
we can choose $\widehat{G}=G$ in \cite[Lemma~3.3]{PS14} and $\wt{V}(x) = 1$ 
in Theorem~\ref{thm: drift}. In that case, 
both results essentially simplify to Corollary \ref{thm: was_mith}. 
The results become entirely different when such a 
bound is not available or too rough. 
In our estimate, one then needs a non-trivial Lyapunov 
function for $\wt P$ and a uniform upper bound 
on $W(\delta_x P,\delta_x \wt P)/\wt{V}(x)$. 
To apply their estimate, one needs a uniform 
bound on $W(\delta_x P,\delta_x \wt P)$ for all  $x \in \widehat{G}$. 
In addition, a bound on $\pi(G\setminus \widehat{G})$,
Lyapunov functions and estimates of the exit probabilities 
from $\widehat{G}$ of both Markov chains need to be available.
Finally, while \cite[Lemma~3.3]{PS14}
requires slightly more regularity on the Lyapunov function, 
contractivity of the unperturbed 
transition kernel
$P$ (with $C=1$) is not 
needed on the whole state space but only on $\widehat{G}$.
 
\end{rem}

\begin{table}[htb]
\label{table: comp}
\centering
\caption{Comparison of the Wasserstein perturbation 
bound of \cite[Lemma~3.3]{PS14} and Theorem~3.1. 
Here  $\rho,\delta \in [0,1)$, $L,c_p,C,D \in (0,\infty)$, 
$V\colon G\to [0,\infty)$, $\wt V\colon G\to [1,\infty)$ and
$
 E(x) = \int_G d(x,y)\dint \pi(y).
$
}

\begin{tabular}{|c|c|c|}
\hline
         \begin{tabular}[c]{@{}c@{}} 
	\\ \\ \\
              \end{tabular}           & 
                 \begin{tabular}[c]{@{}c@{}} 
		Assumptions of \\
                \cite[Lemma 3.3]{PS14}
              \end{tabular} 
                   &  
                      \begin{tabular}[c]{@{}c@{}} 
		Assumptions of \\
                Theorem~\ref{thm: drift}
              \end{tabular} 
                  
\\ 
\hline
\hline
   \begin{tabular}[c]{@{}c@{}}
   \\
		Convergence \\
                property\\
                \\
              \end{tabular}            &  
${ \displaystyle
                            \exists \widehat G \subseteq G
                            \quad \text{s.t.} \quad
                            \sup_{x,y\in \widehat G} 
                            \frac{W(\delta_x P,\delta_y P)}{d(x,y)}\leq \rho
                            }$

              &                 $\tau(P^n)\leq C\rho^n$
              \\
              \hline
   \begin{tabular}[c]{@{}c@{}} 
   \\  Lyapunov function   \\
   \\
              \end{tabular} 
&   
   \begin{tabular}[c]{@{}l@{}} 
		$PV(x)\leq \delta V(x) + L$ \\
                $\wt PV(x)\leq \delta V(x) + L$ 
              \end{tabular}
          &     $\wt P\wt V(x)\leq \delta \wt V(x) + L$         \\
          \hline
             \begin{tabular}[c]{@{}c@{}} 
   \\  Drift regularity    \\
   \\ \\
              \end{tabular} 
          & 
   \begin{tabular}[c]{@{}l@{}} 
		$\mathbb{E}[V(X_{n+1}\mid X_n=x, X_{n+1}\not \in \widehat G)] \leq C$ \\
$\mathbb{E}[V(\wt X_{n+1}\mid \wt X_n=x, \wt X_{n+1}\not \in \widehat G)] \leq C$\\
$\exists p\in \widehat G \quad \text{s.t.} \quad d(x,p)\leq V(x)+c_p$  
\end{tabular}&         ---     \\
\hline
             \begin{tabular}[c]{@{}c@{}} 
   \\  Perturbation error    \\
   \\ 
              \end{tabular} 
        & ${
\displaystyle 
\widehat\gamma:=\sup_{x\in \widehat G} W(\delta_x P,\delta_x \wt P)}$  
			   & ${
 			   \displaystyle 
			   \gamma
			   :=\sup_{x\in G} \frac{W(\delta_x P,\delta_x \wt P)}{\wt V(x)}}$
\\
\hline

\begin{tabular}[c]{@{}c@{}} 
\\
Regularity of $\pi$ \\
\\
\end{tabular}
&  
   \begin{tabular}[c]{@{}c@{}} 
		$\int_{G\setminus \widehat G} V(x) \dint \pi(x)\leq D$ \\
$\pi(G\setminus \widehat G )$  small
\end{tabular}
& ---             \\ \hline \hline
   \begin{tabular}[c]{@{}c@{}} 
   \\
		Conclusion: \\
		Upper bound of\\ 
		$W(\delta_x \wt P^n,\pi)$\\
		\\
\end{tabular}                &  
\begin{tabular}[c]{@{}c@{}} 
		$\rho^n E(x) +\frac{\widehat \gamma }{1-\rho} + $ \\
 $\left(\frac{2L}{1-\delta} + \delta^n (V(x)+D) + c_p\right) \pi(G\setminus \widehat G)$ +\\ 
   $ 2(1-\mathbb{P}[\{X_j\}_{j=1}^{ n-1}\cup \{\wt X_j\}_{j=1}^{n-1}\subseteq \widehat G])
   (C+\frac{L}{1-\delta}+c_p)$ 
\end{tabular}  
&       
\begin{tabular}[c]{@{}c@{}} 
		$C \rho^n E(x) +$ \\
 $
     \frac{C\gamma }{1-\rho} \max\{ \wt V(x),\frac{L}{1-\delta} \}$ 
\end{tabular}    
         \\ \hline
\end{tabular}
\vspace*{2ex}

\end{table}

\subsection{Perturbation bounds for geometrically ergodic Markov chains}\label{sec: pert geom}

In this section, we derive 
general perturbation bounds 
for
geometrically ergodic Markov chains. 
First, we recall some results from \cite{HaMa11}, \cite{MaZhZh13} 
and \cite{RoRo97} which are helpful 
to apply our Wasserstein perturbation bounds in the geometrically ergodic case. 
Then we present the new estimates: 
\begin{itemize}
 \item Corollary~\ref{cor geom1} is an application of Theorem \ref{thm: drift} 
      with Wasserstein distances replaced by $V$-norms of differences between measures.
 \item In Corollary~\ref{cor geom2}, we show that having a Lyapunov function $V$ for $P$ is sufficient 
	for our bounds if the transition kernels $P$ and $\wt P$ are sufficiently close (in a suitable sense).
  \item In Theorem \ref{thm geom3}, we provide a quantitative perturbation bound which still applies 
  if we can only control the total variation distance between $P(x,\cdot)$ and $\wt P(x,\cdot)$. 
To measure the perturbation in such a weak sense is new for geometrically ergodic Markov chains. 
\end{itemize}

A transition kernel $P$ 
with stationary distribution $\pi$ 
 is called geometrically ergodic if there is a constant $\rho\in [0,1)$
  and a measurable function $C\colon G \to (0,\infty)$ 
  such that for $\pi$-a.e. $x\in G$ we have
  \[
   \norm{P^n(x,\cdot)-\pi}_{\text{tv}} 
   \leq C(x) \rho^n.
  \]
  For $\phi$-irreducible and aperiodic Markov chains, it is well known
  that
  geometric ergodicity is equivalent to $V$-uniform ergodicity, see \cite[Proposition~2.1]{RoRo97}.
  Namely, if $P$ is geometrically ergodic, then
  there exists a $\pi$-a.e. finite measurable function $V \colon G \to [1,\infty]$
  with finite moments with respect to $\pi$ and there are constants $\rho\in [0,1)$
  and $C\in(0,\infty)$ such that
  \[
   \norm{P^n(x,\cdot)-\pi}_V 
   := \sup_{\abs{f}\leq V} \abs{\int_G f(y)(P^n(x,\dint y)-\pi(\dint y))}
   \leq C V(x) \rho^n, \quad x\in G, n\in \N.
  \]
  Thus
  \begin{equation}\label{geom_bound}
   \sup_{x\in G} \frac{\norm{P^n(x,\cdot)-\pi}_V}{V(x)} \leq C \rho^n.
  \end{equation}
  The following result establishes 
  the connection between $V$-norms and certain Wasserstein distances.
      It is basically due to Hairer and Mattingly \cite{HaMa11}, 
  see also 
  \cite{MaZhZh13}.  
  \begin{lem}\label{lem dist}
  Assume that $V$ is lower semi-continuous on $G$. 
  For $x,y\in G$, let us define the metric
   \[
    d_V(x,y) = (V(x)+V(y))\mathbf{1}_{x\not =y} = \begin{cases}
                 V(x)+V(y) 	& x\not = y\\
                 0 		& x=y.
               \end{cases}
   \]
 Then, for any $\mu,\nu \in \mathcal{P}$ 
   we have
   \begin{equation} \label{eq: KR_dv}
    \norm{\mu-\nu}_{V} = W_{d_V}(\mu,\nu),    
   \end{equation}
 where $W_{d_V}$ denotes the Wasserstein distance based on the metric $d_V$.
   \end{lem}

Lower semi-continuity of $V$ implies lower semi-continuity of $d_V$, 
which leads to the duality formula \eqref{KRduality} 
by \cite[Theorem~1.14]{Vi03}. 
We thus impose the standing assumption of 
lower semi-continuity of $V$ whenever we speak of $V$-uniform ergodicity in the following.  
In principle, this requirement can be removed and \eqref{eq: KR_dv} remains true, 
but we do not go into further detail in that direction. In applications, this is typically not restrictive since $V$ is continuous anyway.

   By similar arguments as in the proof of \cite[Theorem~1.1]{MaZhZh13}
   we observe that \eqref{geom_bound} implies 
	a suitable upper bound on 
	\[
	\tau_V(P)= 
	\sup_{x,y\in G,\, x\not = y} \frac{W_{d_V}(\delta_x P,\delta_y P)}{d_V(x,y)}
	= 	
	\sup_{x,y\in G,\, x\not = y} \frac{\norm{P(x,\cdot)-P(y,\cdot)}_V}{V(x)+V(y)}.
	\]
	\begin{lem}  \label{lem: V_unif_contr}
  If \eqref{geom_bound} is satisfied for the transition kernel $P$, then
	$
	\tau_V(P^n) \leq C \rho^n.
	$
	\end{lem}
	\begin{proof}
For any positive real numbers $a_1,a_2,b_1,b_2$ we have the following elementary inequality
\begin{equation}\label{eq: elementary}
\frac{a_1+a_2}{b_1+b_2}\leq \max\left\{\frac{a_1}{b_1},\frac{a_2}{b_2}  \right\}.
\end{equation}
%
By \eqref{eq: elementary} we obtain
\begin{align*}
\tau_V(P^n) & 
=\sup_{x,y\in G,\,x \neq y} 
\frac{ W_{d_V}( \delta_x P^n,\delta_y P^n) }{d_V(x,y)} 
\leq 
\sup_{x,y\in G,\,x \neq y} \frac{\| P^n(x,\cdot)- \pi\|_V+\|P^n(y,\cdot)-\pi\|_V }{V(x)+V(y)}\\
&\leq\sup_{x,y\in G} \max\left\{  \frac{\|  P^n (x,\cdot)-\pi\|_V }{V(x)},
\frac{\|  P^n(y,\cdot) -\pi\|_V }{V(y)}  \right\}
= \sup_{x\in G} \frac{\| P^n(x,\cdot) -\pi\|_V }{V(x)}.
\end{align*}
Now, by using \eqref{geom_bound} we obtain the assertion.
\end{proof}
The lemmas above and Theorem~\ref{thm: drift} lead to the following new
perturbation bound for geometrically ergodic Markov chains.
\begin{cor}\label{cor geom1}
    Let $P$ be $V$-uniformly ergodic, i.e.,
there are constants $\rho\in [0,1)$
  and $C\in(0,\infty)$ such that
  \[
   \norm{P^n(x,\cdot)-\pi}_V 
   \leq C V(x) \rho^n, \quad x\in G, n\in \N.
  \]
  We also assume that there are numbers $\delta\in (0,1)$ 
 	and $L\in(0,\infty)$ and a 
 	measurable Lyapunov function $\wt V:G \rightarrow [1,\infty)$ of $\wt P$
 	such that 
  \begin{equation}  \label{eq: geometric_erg_Luyapunov}
    (\wt P \wt V)(x) \leq \delta \wt V(x)+ L.
  \end{equation}
  Let 
  \[
   \gamma = \sup_{x\in G} \frac{\norm{P(x,\cdot)-\widetilde P(x,\cdot)}_V}{\wt V(x)}
  \qquad  \mbox{and} \qquad 
   \kappa=\max\left\{\wt p_0 (\wt V), \frac{L}{1-\delta} \right\}
  \]
  with $\wt{p_0}(\wt V) = \int_G \wt V(x)\, \dint {\wt p_0}(x)$.
  Then
  \begin{equation} 
    \norm{p_n-\widetilde p_n}_V \leq C\left( \rho^n \norm{p_0-\widetilde p_0}_V + (1- \rho^n)
     \frac{\gamma \kappa}{1-\rho} \right).
\end{equation}
      \end{cor}

\begin{rem}
In \cite[Theorem~3.1]{ShSt00}, a related perturbation bound is proven. 
The convergence property of the unperturbed transition kernel
is slightly weaker than our $V$-uniform ergodicity, but
also based on a kind of Lyapunov function. 
More restrictively, there it is assumed that the difference of $P^n$ and $\wt P^n$
for all $n>0$ can be controlled. In addition,
the perturbation error is measured with a weight given 
by the same Lyapunov function as in the convergence property of $P$,
but by taking a supremum over a subset of test functions.
With our approach we can take the supremum over all test functions and obtain
similar estimates by setting $p_0 = \pi$.
\end{rem}

The next corollary demonstrates how the Lyapunov function of $\wt P$ 
can be replaced by a Lyapunov function of $P$, 
provided that the distance between the transition kernels is sufficiently small.
Notice that assuming the existence of 
a Lyapunov function of $P$ 
in addition to the $V$-uniform ergodicity 
is a definition of constants rather than an additional requirement, 
see, e.g., \cite{Ba05}.

\begin{cor}\label{cor geom2}
    Let $P$ be $V$-uniformly ergodic, i.e.,
there are constants $\rho\in [0,1)$
  and $C\in(0,\infty)$ such that
  \[
   \norm{P^n(x,\cdot)-\pi}_V 
   \leq C V(x) \rho^n, \quad x\in G, n\in \N.
  \]
	Moreover, $V\colon G\to [1,\infty)$ is a measurable Lyapunov function of $P$, such that
  \begin{equation} 
    (P  V)(x) \leq \delta  V(x)+ L
  \end{equation}
  with constants $\delta \in (0,1)$ and $L \in (0,\infty)$. 
	Let 
  \[
   \gamma = \sup_{x\in G} \frac{\norm{P(x,\cdot)-\widetilde P(x,\cdot)}_V}{ V(x)}
  \qquad  \mbox{and} \qquad 
   \kappa=\max\left\{\wt p_0 ( V), \frac{L}{1-\delta-\gamma} \right\}
  \]
  with $\wt{p_0}(V) = \int_G V(x)\, \dint {\wt p_0}(x)$.
  If $\gamma+\delta <1$, then
  \begin{equation} 
    \norm{p_n-\widetilde p_n}_V \leq C\left( \rho^n \norm{p_0-\widetilde p_0}_V + (1- \rho^n)
     \frac{\gamma \kappa}{1-\rho} \right).
\end{equation}
      \end{cor}
\begin{proof}
It suffices to show that
\begin{equation}\label{drifttilde}
 (\wt P V)(x) \leq (\delta+\gamma ) V(x)+ L
\end{equation}
and then to apply Corollary \ref{cor geom1}. We have
\[
((\wt P -P)V)(x) \leq |((\wt P -P)V)(x) | \leq \norm{\wt P(x,\cdot)-P(x,\cdot)}_V \leq \gamma \,V(x)
\]
which implies \eqref{drifttilde}. 
The assertion follows by the assumption that $\delta+\gamma<1$ and an application of Corollary~\ref{cor geom1}.
\end{proof}

\begin{rem}
For discrete state spaces and under the requirement $p_0=\wt p_0$,
a result similar to the previous corollary is obtained in 
\cite[Theorem~3, Corollary 3]{KaGo13}. 
The authors of \cite{KaGo13} replace our constant $\kappa$ by $\max_{ 0\leq i \leq n} \wt p_i(V)$.
This we could do as well, see the proof of Theorem \ref{thm: drift}. 
\end{rem}

In the perturbation bound of Corollary~\ref{cor geom1}, the function $V$ plays two roles. 
In its first role, $V$ appears in the $V$-uniform ergodicity condition 
and thus is used to quantify convergence of $P$. 
In its second role, $V$ appears in the constant $\gamma$, with which we compare $P$ and $\wt P$, as well as
in the definition of the distance between
$p_n$ and $\wt p_n$. 
We can interpret $\gamma$ of Corollary~\ref{cor geom1}
as an operator norm of $P-\wt P$.
To this end, let $B_V$ be the set of all 
measurable functions $f\colon G\to \R$ 
with finite 
\begin{equation} \label{eq: f_V_norm}
  \abs{f}_V :=\sup_{x\in G} \frac{\abs{f(x)}}{V(x)},
\end{equation}
which means
\[
  B_V = \left\{ f\colon G \to \R \mid \abs{f}_V 
  <\infty  \right\}.
 \]
 It is easily seen that $(B_V,\abs{\cdot}_V)$ is a normed linear space.
 In the setting of Corollary~\ref{cor geom1}, we have
 \begin{equation} 
  \left \VERT P-\wt P \right \VERT_{ B_V\to B_{\wt V}} := \sup_{\abs{f}_V\leq 1} \abs{(P-\wt P)f}_{\wt V} 
	= \gamma.
  \end{equation}
 In Corollary~\ref{cor geom2}, the more restrictive case 
 $V=\wt V$ is considered.
 The corresponding operator norm $\VERT P-\wt P \VERT_{ B_V\to B_V}$ 
 appears in classical perturbation theory for Markov chains, 
 see \cite{Ka86,KaGo13}. 
 But as discussed in \cite[p.~1126]{ShSt00} and \cite{FHL13} it might be too restrictive
 to measure the perturbation with this operator norm for $V=\wt V$.

By relying, e.g., on \cite[Proposition~2]{Ma04} we have some flexibility in the choice of $V$. 
There it is shown that, for $r \in (0,1)$, $V$-uniform ergodicity implies $V^r$-uniform ergodicity. 
This leads to less favorable constants in the $V^r$-uniform ergodicity of $P$, 
but can relax the requirements on the similarity of $P$ and $\wt P$.
Namely, with a Lyapunov function $\wt V$ of $\wt P$ we can apply 
Corollary~\ref{cor geom1} 
 with a $V^{r}$-uniformly ergodic $P$ and $\gamma=
 \VERT P-\wt P\VERT_{ B_{V^{r}}\to B_{\wt V}}$.
 
Unfortunately, this approach breaks down for $r=0$. To see this, notice that $V^r$-uniform ergodicity 
with $r=0$ is just uniform ergodicity which is not implied by geometric ergodicity. 
The next theorem overcomes this limitation by separating the two roles of the function $V$ 
in the previous perturbation bounds. Roughly,
we set $V=1$ in the sense that we measure the 
distances between $P$ 
and $\wt P$ as well as between $p_n$ and $\wt p_n$ in the total variation distance. 
At the same time, 
we set $V=\wt V$ in the sense that 
we assume $P$ is $\wt V$-uniformly ergodic with 
Lyapunov function $\wt V$. 

\begin{thm}\label{thm geom3}
    Let $P$ be $\wt V$-uniformly ergodic, i.e.,
there are constants $\rho\in [0,1)$
  and 
  $C\in(0,\infty)$ 
  such that
  \[
   \norm{P^n(x,\cdot)-\pi}_{\wt V} 
   \leq C \wt V(x) \rho^n, \quad x\in G, n\in \N.
  \]
Moreover, $\wt V\colon G\to [1,\infty)$ 
is a measurable Lyapunov function of $\wt P$ and $P$, such that
  \begin{equation*} 
    (\wt P  \wt V)(x) \leq \delta  \wt V(x)+ L,
    \qquad \text{and} \qquad ( P \wt V)(x) \leq  \wt V(x)+ L,
  \end{equation*}
  with constants $\delta \in (0,1)$ and $L \in (0,\infty)$. 
	Let 
  \begin{equation}\label{gamma_FHL}
   \gamma = \sup_{x\in G} \frac{\norm{P(x,\cdot)-\widetilde P(x,\cdot)}_{{\rm tv}}}
   { \wt V(x)}
  \qquad  \mbox{and} \qquad 
   \kappa=\max\left\{\wt p_0 (\wt V), \frac{L}{1-\delta} \right\}
  \end{equation}
  with $\wt{p_0}(\wt V) = \int_G \wt V(x)\, \dint {\wt p_0}(x)$.
  Then, for $\gamma \in (0,\exp(-1))$ we have
  \begin{equation} \label{eq: Keller_bound}
    \norm{p_n-\widetilde p_n}_{{\rm tv}} 
    \leq C \rho^n \norm{p_0-\widetilde p_0}_{\wt V} + 
     \frac{\kappa\,\exp(1)}{1-\rho}\,
     (2C(L+1))^{\log(\gamma^{-1})^{-1}}\,\gamma
     \log(\gamma^{-1}).
\end{equation}
     \end{thm}
 \begin{proof}
 From the proof of Theorem~\ref{thm: was_mith} we know that
  \begin{equation*} 
    \norm{ \widetilde p_n - p_n}_{\text{tv}} 
     \leq \norm{(\widetilde p_0 - p_0 )P^n}_{\text{tv}}
     + \sum_{i=0}^{n-1} \norm{\widetilde p_i (\widetilde P-P)P^{n-i-1}}_{\text{tv}},\quad n\in \mathbb{N}.
  \end{equation*}
  By Lemma~\ref{lem: V_unif_contr}, we have
  \[
   \norm{(\widetilde p_0 - p_0 )P^n}_{\text{tv}} 
   \leq \norm{(\widetilde p_0 - p_0 )P^n}_{\wt V}
   \leq C\rho^n \norm{\widetilde p_0 - p_0}_{\wt V}.
  \]
Fix a real number $r \in (0,1)$ and let $s=1-r$.
By considering \eqref{eq: tv_dobrushin_erg_coef} one can see that
$\tau_1(P)\leq 1$. 
This leads to
\begin{align*}
\norm{\widetilde p_i (\widetilde P-P)P^{n-i-1}}_{\text{tv}}
 & \leq \norm{\widetilde p_i (\widetilde P-P)P^{n-i-1}}_{\text{tv}}^{r} 
\norm{\widetilde p_i (\widetilde P-P)P^{n-i-1}}_{\wt V}^{s}\\
& \leq \norm{\widetilde p_i (\widetilde P-P)}_{\text{tv}}^{r}
\norm{\widetilde p_i (\widetilde P-P)}_{\wt V}^{s} \tau_{\wt V}(P^{n-i-1})^{s}.
\end{align*}
We also have
\begin{align*}
 \norm{\widetilde p_i (\widetilde P-P)}_{\text{tv}}
& \leq \int_G \norm{\delta_x P -\delta_x \wt P}_{\text{tv}} \dint \wt p_i(x)
 \leq \gamma \int_G \wt V(x)\,\dint \wt p_i(x),\\ 
 \norm{\widetilde p_i (\widetilde P-P)}_{\wt V}
 & \leq \int_G W_{d_{\wt V}}(\delta_x P,\delta_x \wt P)\, \dint \wt p_i(x)
 \leq \sup_{x\in G} \frac{W_{d_{\wt V}}(\delta_x P,\delta_x \wt P)}{\wt V(x)} 
 \int_G \wt V(x)\, \dint\wt p_i(x). 
 \end{align*}
Moreover, for $i \geq 0$ we obtain
\begin{align*}
 \int_G \wt V(x) \,\dint \wt p_i(x) 
 = \int_G \wt P^i \wt V(x)\, \dint \wt p_0(x)
 \leq \delta^i \wt p_0(\wt V) + \frac{L (1-\delta^i)}{(1-\delta)} 
 \leq \kappa,
 \end{align*}
and, by 
\begin{align*}
 W_{d_{\wt V}}(\delta_x P,\delta_x \wt P) 
& = \inf_{\xi \in M(\delta_x P,\delta_x \wt P)} 
\int_G \int_G (\wt V(z)+\wt V(y))\mathbf{1}_{z\not = y} \, \dint \xi(y,z)\\
& \leq P\wt V(x)+\wt P\wt V(x) \leq (1+\delta) \wt V(x) +2L,
\end{align*}
we have
\begin{align*}
 \sup_{x\in G} \frac{ W_{d_{\wt V}}(\delta_x P,\delta_x \wt P)}{\wt V(x)} \leq 2(L+1).
\end{align*}
   Then
  \begin{align*}
 \norm{ \widetilde p_n - p_n}_{\text{tv}} 
 &  \leq C\rho^n \norm{\widetilde p_0 - p_0}_{\wt V} 
 +2^{s}(L+1)^{s} \gamma^{r} \kappa 
    \sum_{i=0}^{n-1} \tau_{\wt V}(P^i)^{s}.
  \end{align*}
Finally, by Lemma~\ref{lem: V_unif_contr}
we obtain
\[
 \sum_{i=0}^{n-1} \tau_{\wt V}(P^i)^{s} 
 \leq \frac{C^{s}(1-\rho^{ns})}{1-\rho^{s}}
 \leq  \frac{C^{s}}{1-\rho^{s}}
 \leq  \frac{C^{s}}{s(1-\rho)}.
\]
For $\gamma\in (0,\exp(-1))$, 
we can choose the numbers $r = 1+\log( \gamma)^{-1}$ 
and $s = \log(\gamma^{-1})^{-1}$. 
This yields $\gamma^{r} = \exp(1) \gamma$ and the proof is complete.
\end{proof}

\begin{rem} \label{rem: thm geom3}
Let $\wt \pi \in \mathcal{P}$ be a stationary distribution of $\wt P$. 
Notice that by the assumption that $\wt V$ is Lyapunov function of $\wt P$
and \cite[Proposition~4.24]{Ha06} it follows that $\wt \pi(\wt V) \leq L/(1-\delta)$.
Further, by the $\wt V$-uniform ergodicity of $P$ we also know that $\pi(\wt V)$
is finite. Thus,
\[
 \norm{\pi-\wt \pi}_{\wt V} \leq \pi(\wt V) + \wt \pi(\wt V) < \infty.
\]
Now, by Theorem~\ref{thm geom3} we can bound $\norm{\pi-\wt \pi}_{\text{tv}}$
with $p_0=\pi$, $\wt p_0 = \wt \pi$ and by letting $n \to \infty$.
We obtain
\begin{equation}  \label{eq: geom3_bound}
   \norm{\pi -\wt \pi}_{\text{tv}} 
   \leq      \frac{L\,(2C(L+1))^{\log (\gamma^{-1})^{-1}}}{(1-\delta)(1-\rho)}\,
     \exp(1)\,\gamma
     \log(\gamma^{-1}).
 \end{equation} 
\end{rem}

\begin{rem} \label{rem: thm geom4}
Let us comment on the dependence of $\gamma$. 
In Section \ref{sec: Langevin}, we apply Theorem~\ref{thm geom3}
combined with \eqref{eq: geom3_bound}
in a setting 
where we have $\gamma \leq K\cdot \log(N)/N$ for a constant $K\geq1$ and 
some parameter $N\in \mathbb{N}$ of the perturbed transition kernel. 
For $\eps \in (0,1)$ and any $N> (K/\eps)^{1/(1-\eps)}$ we have $\gamma <\exp(-1)$.
Then, with some simple calculations, we obtain for $p_0=\wt p_0$ and $N>6 K^{3/2}$ 
the bound
 \begin{equation*}
   \max\{\norm{p_n-\wt p_n}_{\text{tv}}, \norm{\pi -\wt \pi}_{\text{tv}}\} 
 \leq      \frac{3\kappa\,(2C(L+1))^{2/\log( N)}}{1-\rho}\cdot\frac{K\log(N)^2}{N}.
 \end{equation*}
\end{rem}

\begin{rem}\label{rem: Kell_Liv}
In the setting of Theorem~\ref{thm geom3}, 
we can also interpret $\gamma$ as an operator norm.
Namely,
 \begin{equation} \label{eq: Kell_Liv}
   \left \VERT P - \wt P \right \VERT_{B_1 \to B_{\wt V}}  = \sup_{\abs{f}_1 \leq 1} \abs{(P - \wt P)f}_{\wt  V} 
  = \gamma.
  \end{equation}
 Here the subscript ``$1$'' in $\abs{f}_1$ indicates $V(x)=1$ for all $x\in G$, see \eqref{eq: f_V_norm}. 
 For $\varepsilon_0>0$ and a family of perturbations $(\wt P_\varepsilon)_{\abs{\eps} \leq \eps_0 } $ let 
 $\gamma = \VERT P - \wt P_\varepsilon \VERT_{B_1\to B_{\wt V}} \to 0$ 
 for $\varepsilon \to 0$. 
 This condition appears in \cite[Theorem~1, condition (2)]{FHL13} and is 
 an assumption introduced by Keller and Liverani, see \cite{KeLi99}. 

\end{rem}

\section{Applications}\label{sec: appl}

We illustrate our perturbation bounds in three different settings. 
We begin with studying an autoregressive process also 
considered in \cite{FHL13}. 
After this, we show quantitative perturbation 
bounds for approximate versions of two prominent MCMC algorithms, namely 
the Metropolis-Hastings
and stochastic Langevin algorithms. 

\subsection{Autoregressive process} \label{subsec: auto_regr_proc}

Let $G=\R$ and assume that $(X_n)_{n\in \N_0}$ is the autoregressive
model defined by
\begin{equation}
\label{eq: AR1}
 X_{n} = \alpha X_{n-1} + Z_n, \quad n\in \N.
\end{equation}
Here $X_0$ is an $\R$-valued random variable, $\alpha\in (-1,1)$ and
$(Z_n)_{n\in\N}$ is an i.i.d. sequence of random variables, independent of $X_0$.
We also assume that the distribution of $Z_1$, say $\mu$, 
admits a first moment.
It is easily seen that $(X_n)_{n\in\N_0}$ is a Markov chain with
transition kernel
\[
 P_\alpha(x,A) = \int_{\R} \mathbf{1}_A(\alpha x + y) \,\dint \mu(y),
\]
and it is well known that there exists a stationary distribution, say $\pi_\alpha$,
of $P_\alpha$.

Now, let the transition kernel $P_{\wt \alpha}$ with $\wt \alpha\in (-1,1)$ 
be an approximation of $P_\alpha$. 
For $x,y\in G$, let us consider 
the metric which is given by the absolute difference, i.e., 
$d(x,y)=\abs{x-y}$.
We assume that 
$\abs{\alpha-\wt \alpha}$ is small and 
study the Wasserstein distance, based on $d$, of $p_0 P_\alpha ^n$
and $\wt p_0 P_{\wt \alpha}^n$ with two probability measures
$p_0$ and $\wt p_0$ on $(\R,\mathcal{B}(\R))$.

We intend to apply Theorem~\ref{thm: drift}. Notice that
for 
$\wt V\colon \R \to [1,\infty)$ with $\wt V(x)=1+\abs{x}$ we have 
\begin{align*}
 P_{\wt \alpha} \wt V (x) 
 & \leq  \abs{\wt \alpha} \wt V(x) + 1-\abs{\wt \alpha}+\E \abs{Z_1}
\end{align*}
which guarantees that condition \eqref{eq: drift_cond} is satisfied with 
$\delta=\abs{\wt \alpha}$ and $L=1-\abs{\wt \alpha}+\E\abs{Z_1}$.
Furthermore 
\begin{align*}
 W(\delta_x P_\alpha,\delta_y P_\alpha ) 
\leq \int_{\R} \abs{\alpha x- z - \alpha y + z} \, \dint \mu(z) 
 \leq \abs{\alpha} \abs{x-y} = \abs{\alpha} d(x,y),
\end{align*}
leads to $\tau(P_\alpha^n) \leq \abs{\alpha}^n$.
Similarly, 
one obtains
\[
 W(\delta_x P_\alpha,\delta_x P_{\wt \alpha})
 \leq \int_{\R} \abs{\alpha x- z - \wt \alpha x + z} \, \dint \mu(z)
 \leq  \abs{x} \abs{\alpha-\wt \alpha}
\]
which implies that
\[
\sup_{x\in \R} \frac{ W(\delta_x P_\alpha,\delta_x P_{\wt \alpha})}{\wt V(x)} \leq \abs{\alpha-\wt \alpha}.
\]
We set
\[
 \kappa = 
 1+\max\left\{ \int_{\R} \abs{x} \dint \wt p_0( x), \frac{\E\abs{Z_1}}{1-\abs{\wt \alpha}}    \right\}
\]
and $p_{\alpha,n}= p_0 P^n_\alpha$, $\wt p_{\wt \alpha,n}= \wt p_0 P^n_{\wt \alpha}$.
Then, inequality \eqref{eq: wass_est_d} of Theorem~\ref{thm: drift} 
gives
\begin{equation} \label{eq: ar_wass_est1}
  W(p_{\alpha,n},\wt p_{\wt \alpha,n}) \leq \abs{\alpha}^n W(p_{0},\wt p_{0}) + 
 \abs{\alpha-\wt \alpha} \frac{(1-\abs{\alpha}^n)\,
 \kappa}{1-\abs{\alpha}},
\end{equation}
and for $p_{0}=\wt p_{0}$ we have
\begin{equation} \label{eq: ar_wass_est2}
  W(p_{\alpha,n},\wt p_{\wt \alpha,n}) 
 \leq  \abs{\alpha-\wt \alpha} \frac{ (1-\abs{\alpha}^n)\,\kappa}{1-\abs{\alpha}}.
\end{equation}
From the previous two inequalities one can see that
if $\wt \alpha$ is sufficiently close to $\alpha$, then the distance of the distribution $p_{\alpha,n}$
and $\wt p_{\wt \alpha,n}$ is small. 
Let us emphasize here that 
we provide an explicit estimate rather than an asymptotic statement. 

Note that 
by \cite[Proposition~4.24]{Ha06} 
and the fact that $P_\beta g (x) \leq \abs{\beta} g(x)+\mathbb{E}\abs{Z_1}$\
with $g(x)=\abs{x}$ and $\beta\in \{\alpha,\wt \alpha\}$ we obtain
$
 \int_\mathbb{R} \abs{x} \, \dint \pi_\beta(x)<\infty, 
$
which leads to a finite $W(\pi_\alpha,\pi_{\wt \alpha})$.
As a consequence we obtain 
for the stationary distributions of $P_\alpha$ and $P_{\wt \alpha}$ 
by estimate \eqref{eq: dist_stat_dist2} that
\begin{equation} \label{eq: ar_diff_stat_dist}
  W(\pi_\alpha,\pi_{\wt \alpha}) 
 \leq \abs{\alpha-\wt \alpha} \frac{1-\abs{\wt \alpha} 
    + \mathbb{E}\abs{Z_1}}{(1-\abs{\alpha})(1-\abs{\wt \alpha})}.
\end{equation}
The dependence on 
$\abs{\alpha-\wt \alpha}$ in the previous inequality cannot be improved in general.
To see this, let us assume that 
$X_{0,\alpha}$ and $X_{0,\wt \alpha}$ are real-valued random variables 
with distribution $\pi_\alpha$ and $\pi_{\wt \alpha}$, respectively. Then, because of the stationarity 
we have that $X_{1,\alpha} = \alpha X_{0,\alpha} + Z_1$ 
and $X_{1,\wt \alpha} = \wt \alpha X_{0,\wt \alpha} + Z_1$
are also distributed according to $\pi_\alpha$ and $\pi_{\wt \alpha}$, respectively.
Thus
\[
 \mathbb{E} X_{0,\alpha} = \frac{\mathbb{E}Z_1 }{1-\alpha},\qquad 
  \mathbb{E} X_{0,\wt \alpha} = \frac{\mathbb{E}Z_1}{1-\wt \alpha}.
\]
Now, for $g \colon \R \to \R$ with $g(x)=x$, we have $\norm{g}_{\text{Lip}} \leq 1$ 
and thus
\begin{align*}
  W(\pi_\alpha,\pi_{\wt \alpha}) 
& = \sup_{\norm{f}_{\text{Lip}} \leq 1} \abs{\int_G f(x) (\dint \pi_\alpha(x) -\dint \pi_{\wt \alpha}(x))}\\
& \geq \abs{\int_G x\, (\dint \pi_\alpha(x) -\dint \pi_{\wt \alpha}(x))} 
= \abs{\mathbb{E} X_{0,\alpha}-\mathbb{E}X_{0,\wt \alpha}} \\
&= \abs{\alpha-\wt \alpha} \frac{\abs{\mathbb{E}Z_1}}{\abs{1-\alpha}\abs{1-\wt \alpha}}.
\end{align*}
Hence, whenever $\mathbb{E}Z_1 \not = 0$ 
we have a non-trivial lower bound with the same dependence on 
$\abs{\alpha-\wt \alpha}$ as in the upper bound of \eqref{eq: ar_diff_stat_dist}. 
This fact shows that we cannot improve the upper bound.

 Let us now discuss the application of Corollary~\ref{cor geom2} and Theorem~\ref{thm geom3}.
 Under the additional assumption that $\mu$, the distribution of $Z_1$, 
 has a Lebesgue density $h$, it is shown in
\cite[Section~4]{GuHeLe12}
 that the autoregressive model \eqref{eq: AR1} is also $\wt V$-uniformly
 ergodic. 
Precisely, there is a constant $C\geq1$ such that
 \[
  \norm{P_\alpha^n(x,\cdot)-\pi_\alpha}_{\rm tv} \leq C\abs{\alpha}^n \wt V(x).
 \]
Moreover, from \cite[Example~1]{FHL13}
we know that 
 \[
  \sup_{x\in \mathbb{R}} \frac{\norm{P_\alpha(x,\cdot)-P_{\wt \alpha}(x,\cdot)}_{\wt V}}{\wt V(x)}
 \]
 does not go to $0$ when $\wt \alpha \downarrow \alpha$. Hence, 
 Corollary~\ref{cor geom2} cannot quantify for small $\vert \wt \alpha - \alpha \vert$ 
 whether the $n$th step distributions are close to each other. 
 However, also in \cite[Example~1]{FHL13} it is proven that
 \[
    \sup_{x\in \mathbb{R}} 
    \frac{\norm{P_\alpha(x,\cdot)-P_{\wt \alpha}(x,\cdot)}_{\rm tv}}{\wt V(x)}
    \to 0 
    \qquad \text{if} \qquad
    \wt \alpha \to \alpha.
 \]
 This indicates that Theorem~\ref{thm geom3} is applicable.
 By assuming in addition that $h$ is \emph{weakly unimodal}\footnote{The function $h\colon \mathbb{R} \to [0,\infty)$
 is called \emph{weakly unimodal} if there exists $s\in \mathbb{R}$ such that $h(x)$ is nondecreasing for $x\in(-\infty,s)$ and nonincreasing for $x\in(s,\infty)$.} and bounded from above by $h_{\max}$, we can quantify the result. Namely,
 \begin{align*}
    \sup_{x\in \mathbb{R}} 
    \frac{\norm{P_\alpha(x,\cdot)-P_{\wt \alpha}(x,\cdot)}_{\rm tv}}{\wt V(x)}
   &  = \sup_{x\in \mathbb{R}}
     \frac{\norm{\mu(\cdot - \alpha x) - \mu(\cdot -\wt \alpha x)}_{\rm tv}}{1+\abs{x}} \\
   & = \sup_{x\in \mathbb{R}} \frac{\int_{\mathbb{R}} \abs{h(z-\alpha x)-h(z-\wt \alpha x)} \dint z }{1+\abs{x}}  
    \leq 2 \abs{\alpha-\wt \alpha} h_{\max}.
 \end{align*}
To see the final estimate, define $F(a) = \int_\mathbb{R} |h(z)-h(z-a)| \dint z$ for $a \in \mathbb{R}$. By unimodality, there exists for any fixed $a \geq 0$ a constant $c$ such that 
$$\int_\mathbb{R} |h(z)-h(z-a)| \dint z = \int_{-\infty}^c h(z)-h(z-a) \dint z  + \int_{c}^\infty h(z-a)-h(z) \dint z .$$ 
The first summand on the right hand side we can bound by 
\[
\int_{-\infty}^c h(z) \dint z -\int_{-\infty}^c h(z-a) \dint z 
= \int_{c-a}^c h(z) \dint z
\leq a\,h_{\max}
\]
and similarly for the second summand. Using that $F(a)=F(-a)$, 
we obtain 
$F(a) \leq 2 |a|\,h_{\max}$.
Finally, by substitution 
we can write
\begin{align*}
\sup_{x\in \mathbb{R}} \frac{\int_\mathbb{R} |h(z-\alpha x)-h(z-\widetilde\alpha x) | \dint z}{1+|x|}
= |\alpha-\wt\alpha|\; \sup_{a \geq 0}  \frac{F(a)}{a+|\alpha-\widetilde\alpha|} \leq 2  |\alpha-\widetilde\alpha|  h_{\max}.
\end{align*}
For simplicity set $p_0= \wt p_0$ and 
assume that $h_{\max}\leq 1$ as well as $\abs{\alpha-\wt \alpha}\in(0,\exp(-1)/2)$.
Then, Theorem~\ref{thm geom3} implies
 \begin{align*}
 \max\{ \norm{p_{\alpha,n}-\wt p_{\wt \alpha,n}}_{\rm tv}, \norm{\pi_\alpha-\pi_{\wt \alpha}}_{\rm tv}\}  
 \leq \frac{\kappa \exp(1)}{1-\abs{\alpha}} (2C(\mathbb{E}\abs{Z_1}+2))
 \abs{\alpha-\wt \alpha} \log(\abs{\alpha-\wt \alpha}^{-1})
 \end{align*}
  which seems to be new. 
 
 \subsection{Approximate Metropolis-Hastings algorithms}
 \label{subsec: appr_Metro}
 
We apply our perturbation results to 
the approximate (or noisy) Metropolis-Hastings algorithms 
analyzed in \cite{AFEB14, BDH14,BaDoHo15,KCW14,MeLeRo15,PS14}.
We assume either that the unperturbed
transition kernel of the 
Metropolis-Hastings
algorithm 
satisfies the Wasserstein ergodicity condition stated in Assumption~\ref{ass: wass_contr} or
is geometrically ergodic.
In particular, we do not assume that 
the transition kernel is 
uniformly ergodic.
Let $\pi$ be a probability distribution on $(G,\mathcal{B}(G))$ 
and assume that we are interested in sampling realizations from this distribution.
Let $Q$ be a transition kernel which serves as the proposal for the Metropolis-Hastings 
algorithm. From \cite[Proposition~1]{Ti98} 
we know that there exists a set $S\subset G\times G$
such that we can define the ``acceptance ratio'' for $(x,y)\in G\times G$ as
\begin{equation} \label{eq: ratio}
  r(x,y) := \begin{cases}
            \frac{\pi(\dint y) Q(y,\dint x)}{\pi(\dint x) Q(x,\dint y)} & (x,y)\in S\\
            0 & \text{otherwise}.
           \end{cases}
\end{equation}
Then, let the acceptance probability be $\alpha(x,y) = \min\{1,r(x,y)\}$.
With this notation the Metropolis-Hastings algorithm defines a transition kernel
\begin{equation}  \label{eq: MH}
  P_{\alpha}(x,\dint y) = Q(x,\dint y) \alpha(x,y) + \delta_x(\dint y)\, s_\alpha(x),
\end{equation}
with 
\[
 s_\alpha(x) = 1-\int_G \alpha(x,y)\, Q(x,\dint y).
 \]
 We provide a step of a Markov chain $(X_n)_{n\in\N_0}$ 
 with transition kernel $P_\alpha$ in algorithmic form.
\begin{alg}\label{alg MH}
A single transition from $X_n$ to $X_{n+1}$ 
of the Metropolis-Hastings algorithm works as
follows:
 \begin{enumerate}
 \item Draw a sample $Y\sim Q(X_n,\cdot)$ and $U \sim \mbox{Unif}[0,1]$ independently, 
 call the result $y$ and $u$;
 \item Set $r:=r(X_n,y)$, with the ratio $r(\cdot,\cdot)$ defined in \eqref{eq: ratio};
 \item If $u < r$, then accept the proposal, and set $X_{n+1}:=y$,
 else reject the proposal and set $X_{n+1} := X_n$.
\end{enumerate}
\end{alg}

Now, suppose we are unable to evaluate 
$r(x,y)$, so that 
we are forced to work with an approximation of $\alpha(x,y)$. 
The key idea behind approximate Metropolis-Hastings algorithms 
is to replace $r(x,y)$ by a non-negative 
random variable $R$ with distribution, say $\mu_{x,y,u}$,
depending on $x,y\in G$ and $u\in [0,1]$. For concrete choices of the random variable $R$ 
we refer to 
\cite{AFEB14,BDH14,BaDoHo15,KCW14}.
We present a step of the corresponding Markov chain $(\wt X_n)_{n\in \N}$ 
in algorithmic form.
\begin{alg}\label{alg aMH}
 A single transition from $\wt X_n$ to $\wt X_{n+1}$ works as
follows:
 \begin{enumerate}
 \item Draw a sample $Y\sim Q(\wt X_n,\cdot)$ and $U \sim \mbox{Unif}[0,1]$ independently, 
 call the result $y$ and $u$;
 \item Draw a sample $R\sim \mu_{\wt X_n,y,u}$, 
 call the result $\widetilde r$;
 \item If $u < \widetilde r$, then accept the proposal, and set $\wt X_{n+1}:=y$,
  else reject the proposal and set $\wt X_{n+1} := \wt X_n$.
\end{enumerate}
\end{alg}
The algorithm has acceptance probability
\[
 \wt \alpha (x,y)
 = \mathbb{E} \mathbf{1}_{[0,\min\{1,R\}]}(U)
 = \int_0^1 \int_0^\infty \mathbf{1}_{[0,\min\{1,\widetilde r\}]}(u)\; \dint \mu_{x,y,u}( \widetilde r) \dint u
\]
and the transition kernel of such a Markov chain
is still of the form \eqref{eq: MH}
with $\alpha(x,y)$ substituted by 
$\wt \alpha(x,y)$, i.e.,
it is given by $P_{\wt\alpha}$.  
The following
results hold in the 
slightly more general case where 
$\wt\alpha(x,y)$ is any approximation of the acceptance probability $\alpha(x,y)$.

The next 
lemma provides an estimate for the Wasserstein distance 
between transition kernels of the form \eqref{eq: MH} 
in terms of the acceptance probabilities. 
\begin{lem}  \label{lem: noisy_wass}
Let $Q$ be a transition kernel 
on $(G,\mathcal{B}(G))$ and let $\alpha\colon G \times G \rightarrow [0,1]$ 
and  $\wt \alpha\colon G \times G \rightarrow [0,1]$ be measurable functions. 
By $P_{\alpha}$ and $P_{\wt \alpha}$ we denote 
the transition kernels of the form \eqref{eq: MH} with 
acceptance probabilities $\alpha$ and $\wt \alpha$. 
Then, for all $x\in G$, we have
\[
W(\delta_x P_\alpha, \delta_x P_{\wt\alpha})
\leq \int_G d(x,y)\, \mathcal{E}(x,y) \, Q(x, \dint y)
\]
with $\mathcal{E}(x,y)=|\alpha(x,y)- \wt \alpha(x,y)|$.
\end{lem}

\begin{proof}
By the use of the dual representation of
the Wasserstein distance it follows that
\begin{align*}
 & W(\delta_x P_\alpha, \delta_x P_{\wt \alpha} )
  = \sup_{\norm{f}_{\rm Lip}\leq 1} 
 \abs{\int_G f(y) \left(P_\alpha(x,\dint y) - P_{\wt \alpha}(x,\dint y)\right)}\\
 & = \sup_{\norm{f}_{\rm Lip}\leq 1} 
 \abs{\int_G (f(y)-f(x))(\alpha(x,y)-\wt\alpha(x,y)) Q(x,\dint y)}
 \leq \int_G d(x,y) \mathcal{E}(x,y)Q(x,\dint y).
\end{align*}
\end{proof}
By the previous lemma 
and Theorem~\ref{thm: drift}, we obtain the following Wasserstein
perturbation bound for the approximate Metropolis-Hastings algorithm.
\begin{cor}\label{corMH}
Let $Q$ be a transition kernel 
on $(G,\mathcal{B}(G))$ and 
let $\alpha\colon G \times G \rightarrow [0,1]$ 
and  $\wt \alpha\colon G \times G \rightarrow [0,1]$ be measurable functions. 
By $P_{\alpha}$ and $P_{\wt\alpha}$ we denote 
the transition kernels of the form \eqref{eq: MH} 
with 
acceptance probabilities $\alpha$ and $\wt\alpha$. 
Let the following conditions be satisfied:
 \begin{itemize}
  \item Assumption~\ref{ass: wass_contr} holds for the transition kernel $P_\alpha$, i.e., 
  $\tau(P_\alpha^n) \leq C \rho^n$ for $\rho \in [0,1)$ and $C\in (0,\infty)$.
\item There are numbers $\delta\in (0,1)$, $L\in(0,\infty)$ 
   and a measurable Lyapunov function 
   $\wt V:G \rightarrow [1,\infty)$	of $P_{\wt\alpha}$, i.e., 
	\begin{equation} \label{eq: drift_condMH}
    ( P_{\wt\alpha} \wt V)(x) \leq \delta \wt V(x)+ L.
  	\end{equation}
  \item Let $\mathcal{E}(x,y)= |\alpha(x,y)- \wt\alpha(x,y)|$ and assume that
  \begin{equation} \label{eq: MHgamma}
   \gamma 
   = \sup_{x\in G} \frac{\int_G d(x,y)\, \mathcal{E}(x,y)\, Q(x,\dint y)}{\wt V(x)}
   <\infty.
  \end{equation}
 \end{itemize}
Then, for any $p_0\in \mathcal{P}$
and finite $p_0(\wt V)=\int_G \wt V(x)\dint p_0(x)$
we have
\[
 W(p_0 P_{\alpha}^n,p_0 P_{\wt\alpha}^n) \leq \frac{\gamma \, \kappa\, C(1-\rho^n)}{1-\rho}
\]
where $\kappa =\max\left\{ p_0(\wt V), \frac{L}{1-\delta} \right\}$.
\end{cor}

Let us point out several aspects of condition \eqref{eq: drift_condMH}.
 Recall that \eqref{eq: drift_condMH} is always satisfied with $\wt V(x) = 1$ 
 for all $x\in G$.
 However, in this case it seems more difficult to control $\gamma$. 
 If some additional knowledge in form of a Lyapunov function 
 $V\colon G\to [1,\infty)$ of $P_\alpha$, i.e., $P_\alpha V(x) \leq \delta V(x) + L$
 for some $\delta\in (0,1)$ and $L\in (0,\infty)$, is available, 
 then a non-trivial candidate for $\wt V$ is $V$.
 For sufficiently small
 \[
  \delta_V =  \sup_{z\in G} \int_G \left(\frac{V(y)}{V(z)}+1\right) \mathcal{E}(z,y) Q(z,\dint y) 
 \]
 this is indeed true. Namely, we have
 \begin{align*}
    \abs{(P_\alpha-P_{\wt \alpha}) V  (x)}
  \leq \int_G V(y) \mathcal{E}(x,y) Q(x,\dint y) + V(x) \int_G \mathcal{E}(x,y) Q(x,\dint y) 
  \leq V(x) \delta_V.
 \end{align*}
 Then, $P_{\wt \alpha} V(x) \leq (\delta+\delta_V) V(x) + L$ 
 and whenever $\delta+ \delta_V <1 $ it is clear that condition
 \eqref{eq: drift_condMH} is verified.

 To highlight the usefulness of a non-trivial Lyapunov function, 
 we consider the following scenario which is related to a local perturbation 
 of an independent Metropolis-Hastings algorithm. 
 \begin{ex}
 Let us assume that for $P_\alpha$ 
 Assumption~\ref{ass: wass_contr}, as formulated in Corollary~\ref{corMH}, is satisfied.
 For some probability measure $\mu$ on $(G,\mathcal{B}(G))$ define $Q(x,\cdot)=\mu$ and
 $p_0 = \wt p_0 = \mu$.
 For $\wt G \subseteq G$
  let
  \[
   \wt \alpha(x,y) = \min \{1,\alpha(x,y)+\mathbf{1}_{\wt G}(x)\}.
  \]
  Hence, for $x\in \wt G$ the transition kernel $P_{\wt \alpha}(x,\cdot)$ accepts any 
  proposed state and for $x\not \in \wt G$ 
  we have $P_{\wt \alpha}(x,\cdot) = P_\alpha(x,\cdot)$.
  It is easily seen that $\mathcal{E}(x,y) \leq \mathbf{1}_{\wt G}(x)$.
  For arbitrary $R>0$ and $r\in(0,1)$ set $\wt V(x)= 1+ R\mathbf{1}_{\wt G}(x)$ and note that
  \[
   P_{\wt \alpha} \wt V(x) 
   \leq r \wt V(x) + 1-r+R P_{\wt\alpha}(x,\wt G)
   \leq r \wt V(x) + 1-r+R \mu(\wt G).
  \]
  The last inequality of the previous formula follows by distinguishing the cases $x\in \wt G$ and $x\not\in \wt G$.
  Define $D(\wt G)=\sup_{x\in \wt G} \int_G d(x,y)\mu(\dint y)$ and
  observe
  \begin{align*}
   \kappa & = 1+\frac{R \mu(\wt G)}{1-r}, 
   \qquad \text{and} \qquad
   \gamma \leq \frac{D(\wt G)}{1+R}.
  \end{align*}
  Then, Corollary~\ref{corMH} leads to
  \[
   W(p_0P_\alpha^n,p_0 P_{\wt \alpha}^n) \leq \frac{C}{1-\rho} 
   \left(1+\frac{R \mu(\wt G)}{1-r}\right) \frac{D(\wt G)}{1+R}
  \]
  for arbitrary $R\in(0,\infty)$ and $r\in(0,1)$.
  Under the assumption that $D(\wt G)$ is finite and letting $R\to \infty$ as well as $r\downarrow 0$
  we obtain
  \[
   W(p_0P_\alpha^n,p_0 P_{\wt \alpha}^n) \leq \frac{C \mu(\wt G) D(\wt G) }{1-\rho},
  \]
  which tells us that basically $\mu(\wt G)$ measures the difference of the distributions.
	A small
  perturbation set $\wt G$ with respect to $\mu$, thus implies a small bias. 
In contrast, with the	trivial Lyapunov function $\wt V=1$, and if 
  there is $(x,y)\in \wt G\times G$ such that $\alpha(x,y)=0$, we only obtain
  \[
  \gamma \kappa = D(\wt G) \geq \inf_{x\in G} \int_G d(x,y)\mu(\dint y).
  \] 
	The resulting upper bound on $ W(p_0P_\alpha^n,p_0 P_{\wt \alpha}^n)$ will typically be bounded away from zero regardless of the set $\wt G$.
  \end{ex}

\begin{rem} \label{rem_noisy_metro}
The constant $\gamma$ essentially depends on 
the distance $d(x,y)$ 
and the difference of the acceptance probabilities $\mathcal{E}(x,y)$.
By applying the Cauchy-Schwarz inequality 
to the numerator of $\gamma$, we can separate the two parts, i.e.,
\[
\int_G d(x,y)\, \mathcal{E}(x,y)\, Q(x,\dint y) \leq \left( \int_G d(x,y)^2\, Q(x,\dint y) \cdot
\int_G  \mathcal{E}(x,y)^2\, Q(x,\dint y) 
\right)^{1/2}.
\]
If both integrals remain finite we see that
an appropriate control of  $\mathcal{E}(x,y)$ 
suffices for making the constant $\gamma$ small.
\end{rem}
 
\begin{rem} \label{it: noisy_metr}
By using a Hoeffding-type bound,
in Bardenet et al. \cite[Lemma~3.1.]{BDH14} it is shown that for their
version of the
approximate Metropolis-Hastings algorithm 
with adaptive subsampling 
the approximation error $\mathcal{E}(x,y)$ 
is bounded uniformly in $x$ and $y$ by a constant $s >0$. 
Moreover, $s$ can be chosen arbitrarily small for the implementation of the algorithm. 
\end{rem}

Now we consider the case where the unperturbed transition kernel 
$P_{\alpha}$ is geometrically ergodic.
Motivated by Remark~\ref{it: noisy_metr}, we also 
assume that $\mathcal{E}(x,y)\leq s$
for a sufficiently small number $s>0$.
The following corollary generalizes a
main 
result of Bardenet et al.
\cite[Proposition~3.2]{BDH14} 
to the geometrically ergodic case.

\begin{cor}  \label{cor: metro_geom}
Let $Q$ be a transition kernel 
on $(G,\mathcal{B}(G))$ and let $\alpha\colon G \times G \rightarrow [0,1]$ 
and  $\wt\alpha\colon G \times G \rightarrow [0,1]$ be measurable functions. 
By $P_{\alpha}$ and $P_{\wt\alpha}$ we denote the transition kernels of 
the form \eqref{eq: MH} with acceptance probabilities $\alpha$ and $\wt\alpha$. 
Let the following conditions be satisfied:
 \begin{itemize}
  \item The unperturbed transition kernel $P_\alpha$ is $V$-uniformly ergodic, that is,
	\[
   \norm{P_{\alpha}^n(x,\cdot)-\pi}_V 
   \leq C V(x) \rho^n, \quad x\in G, n\in \N
  \]
for numbers $\rho \in [0,1)$, $C\in (0,\infty)$ and a measurable function $V\colon G \rightarrow [1,\infty)$.  
Moreover, $V$ is a Lyapunov function of $P_{\alpha}$, i.e.,
\begin{equation}\label{eq: LyapunovPalpha}
    ( P_{\alpha} V)(x) \leq \delta  V(x)+ L,
\end{equation}
  for numbers $\delta\in (0,1)$ and $L\in(0,\infty)$.
\item A uniform bound $s >0$ 
on the difference of the acceptance probabilities is given, 
	that is, for all $x,y\in G $, we have 
		\[
	\mathcal{E}(x,y) = |\alpha(x,y)-\wt\alpha(x,y) | \leq s.
	\]
\item The constant $\lambda$ satisfies
\[
\lambda =1 + \sup_{x\in G} \int_G \frac{V(y)}{V(x)} Q(x,\dint y) <\infty.
\]
	\end{itemize}
If $s < (1- \delta)/\lambda$, then, for any $p_0\in \mathcal{P}$ with 
finite 
$
\kappa = \max\left\{ p_0( V) , \frac{L}{1-\delta-\lambda s} \right\}
$
we have
\[
\|p_0 P_{\alpha}^n-p_0 P_{\wt\alpha}^n\|_{V}
\leq  
 \frac{\lambda\, s \, \kappa\,C\, (1-\rho^{ n})}{1-\rho }.
\]
\end{cor}

\begin{proof}
We consider 
the metric $d_V$, defined in Lemma \ref{lem dist}, 
set $V=\wt V$ and use $\mathcal{E}(x,y) \leq s$ so that it is easily seen that
the constant $\gamma$ from Corollary \ref{corMH} satisfies
$\gamma \leq s \lambda$. 
From the proof of Corollary \ref{cor geom2},
we know that $V$ is a Lyapunov function of $P_{\wt\alpha}$ provided that $\gamma+\delta<1$.
Thus, we have
\begin{equation} \label{eq: Lang_Metro}
  P_{\wt \alpha} V(x) \leq (\delta+\lambda s)V(x)+L. 
\end{equation}
Now if $s < (1- \delta)/\lambda$, then  $\delta+\lambda s<1$ and the assertion
follows from Corollary \ref{corMH} by writing the Wasserstein distances 
in terms of $V$-norms as in Section \ref{sec: pert geom}. 
\end{proof}

\begin{rem}
Without $V(x)$ 
in the denominator, i.e., if we had relied on 
Corollary \ref{thm: was_mith} instead of Theorem \ref{thm: drift}, 
the constant
$\lambda$ would often be infinite. Consider the following toy example:
Let $\pi$ be the exponential distribution with density $\exp(-x)$ on $G=[0,\infty)$
and assume that $Q(x,\dint y)$ is a uniform 
proposal with support $[x-1,x+1]$. With $V(x)=\exp(x)$ it is well known that
the Metropolis-Hastings algorithm is $V$-uniformly ergodic, 
see \cite{MeTw96} or \cite[Example~4]{RoRo04}. 
In this example
\begin{align*}
 \lambda \leq 1 + \sup_{x\in [0,\infty)} \int_{x-1}^{x+1} \exp(y-x) \dint y \leq 1+ \exp(1)
\end{align*}
whereas $\int_{x-1}^{x+1} \exp(y) \dint y$ is unbounded in $x$. 
Notice that $\lambda$ only depends on the unperturbed 
Markov chain so that a bound on $\lambda$ can 
be combined with any approximation.
\end{rem}

\begin{rem}
Let $P_{\wt \alpha}$ and $P_{\alpha}$ 
be $\phi$-irreducible and aperiodic.
Then, one can prove under the assumptions
of Corollary~\ref{cor: metro_geom}
that $P_{\wt \alpha}$
is $V$-uniformly ergodic if $s$ is sufficiently small. 
To see this, 
note that by \cite[Theorem~16.0.1]{MeTw09} the 
$V$-uniform ergodicity of $P_{\alpha}$ implies 
that $P_{\alpha}$ satisfies their drift condition (V4). 
By the arguments stated in the proof of Corollary \ref{cor geom2}, 
one obtains that $P_{\wt \alpha}$ also satisfies (V4) 
for sufficiently 
small $s$ and this implies $V$-uniform ergodicity.
In this case, clearly $P_{\wt \alpha}$
 possesses a stationary distribution, say $\wt \pi$,
and
\[
 \norm{\pi -\wt \pi }_V 
 \leq 
 \frac{\lambda\, s \,C}{1-\rho } \cdot 
 \frac{L}{1-\delta-\lambda s} .
\]
The previous inequality follows by \eqref{eq: dist_stat_dist2} and the fact that
\[
 \norm{\pi-\wt \pi}_{V} \leq \pi(V)+\wt \pi(V) <\infty.
\]
Here the finiteness of $\pi(V)$ follows by the $V$-uniform ergodicity of $P$
and $\wt \pi(V) \leq L/(1-\delta-\lambda s)$ follows by \eqref{eq: Lang_Metro} and \cite[Proposition~4.24]{Ha06}.

\end{rem}

\subsection{Noisy Langevin algorithm for Gibbs random fields}\label{sec: Langevin}

An alternative to the Metropolis-Hastings algorithm
is the Langevin algorithm, see \cite{RoTw96b}.
Unfortunately, in its implementation
one needs the
gradient of the density of the target distribution.
To overcome this problem, different approximate Langevin
algorithms have been proposed and studied, see \cite{AhKoWe12,AFEB14,TeThVo14,WeTe11}.

This section is mainly based on Alquier et al. \cite[Section~3.4]{AFEB14}
where a noisy Langevin
algorithm for Gibbs random fields is considered.
We provide a quantitative version of \cite[Theorem~3.2]{AFEB14}.
The setting is as follows. Let $\mathcal{Y}$ be a finite set
and with $M\in \N$ let $y = \{y_1,\dots,y_M\} \in \mathcal{Y}^M$ 
be an observed data set on nodes $\{1,\dots,M\}$ of a certain graph.
The likelihood of $y$ with parameter $\theta \in \R$ is defined by
\[
 \ell(y \,|\, \theta) 
 = \frac{\exp(\theta\, s(y))}{ \sum_{y\in \mathcal{Y}^M} \exp(\theta\, s(y))},
\]
where $s\colon \mathcal{Y}^M \to \R$ is a given statistic.
The density of the posterior distribution with respect to the Lebesgue measure 
on $(\R,\mathcal{B}(\R))$
given the data $y \in \mathcal{Y}^M$ is determined by
\[
 \pi_y(\theta) := \pi(\theta \,|\, y) \propto \ell(y \,|\, \theta) \,p(\theta)
\]
where the prior density $p(\theta)$ is the Lebesgue density 
of the normal distribution $\mathcal{N}(0,\sigma_p^2)$ 
with $\sigma_p>0$.

We consider the Langevin algorithm, 
a first order Euler discretization
of the SDE of the
Langevin diffusion, see \cite{RoTw96b}.
It is given by 
$(X_n)_{n\in \N_0}$
with
\begin{equation}  \label{eq: langevin_rec}
  X_n = X_{n-1} + \frac{\sigma^2}{2} \nabla \log \pi_y(X_{n-1}) + Z_n, \qquad n\in \N.
\end{equation}
Here $X_0$ is a real-valued random variable and $(Z_n)_{n\in \N}$
is an i.i.d. sequence of random variables, independent of $X_0$, with 
$Z_n\sim \mathcal{N}(0,\sigma^2)$ for a parameter $\sigma >0$ 
which can be interpreted as the step size in the discretization 
of the diffusion. 
It is easily seen 
that $(X_n)_{n\in \N_0}$ is a Markov chain with transition kernel
\[
  P_{\sigma}(\theta,A) = \int_\R 
  \mathbf{1}_A\left(\theta+\frac{\sigma^2}{2} \nabla \log \pi_y(\theta) + z\right) 
  \mathcal{N}(0,\sigma^2)(\dint z),\qquad A\in \mathcal{B}(\R).
\]
In general $\pi_y$ is not a stationary distribution of $P_\sigma$, 
but 
there exists a stationary distribution (see Proposition~\ref{prop: Langevin} below),
 say $\pi_\sigma$, which is close to $\pi_y$ depending
on $\sigma$. 
Let 
 $ z(\theta)  = \sum_{y\in \mathcal{Y}^M} \exp(\theta \, s(y))$
then, by the definition of $\pi_y$ we have
\begin{align*}
  \log \pi_y(\theta) & = \theta \, s(y) - \log z(\theta) + \log p(\theta) - 
  \log\left( \int_\R \ell(y\,|\, z) p(z) \dint z \right),\\
  \nabla \log \pi_y(\theta) 
  & = s(y) - \frac{z'(\theta)}{z(\theta)} +\nabla \log p(\theta) \\
  & = s(y) - \frac{\sum_{z\in \mathcal{Y}^M} s(z) \exp(\theta\, s(z))}{\sum_{z\in \mathcal{Y}^M} \exp(\theta\, s(z))}
      - \frac{\theta}{\sigma_p^2} \\
  & = s(y) - \mathbb{E}_{\ell(\cdot \mid \theta)} s(Y) 
      - \frac{\theta}{\sigma_p^2},
\end{align*}
where $Y$ is a random variable on $\mathcal{Y}^M$ 
distributed according the likelihood distribution determined by $\ell(\cdot\,|\, \theta)$. 
We do not have access to the exact value of the mean $\mathbb{E}_{\ell(\cdot \mid \theta)} s(Y) $
since in general we do not know the normalizing constant of the likelihood.
We assume that we can use a Monte Carlo estimate. 
For $N\in \N$
let $(Y_i)_{1\leq i \leq N}$ 
be an i.i.d. sequence of random
variables with $Y_i\sim \ell(\cdot\,|\, \theta)$ independent of $(Z_n)_{n\in\N}$ from \eqref{eq: langevin_rec}.
Then,
$
 \frac{1}{N} \sum_{i=1}^N s(Y_i)
$
is an approximation of $\mathbb{E}_{\ell(\cdot \mid \theta)}s(Y) $
which leads to an estimate of $\nabla \log \pi_y(\theta)$ given by
\[
 \widehat \nabla^{N} \log \pi_y(\theta) 
 := s(y) -  \frac{1}{N} \sum_{i=1}^N s(Y_i) - \frac{\theta}{\sigma_p^2}.
\]
We substitute $\nabla \log \pi_y(\theta)$ by $\widehat \nabla^{N} \log \pi_y(\theta)$
in \eqref{eq: langevin_rec}
and obtain a sequence of random variables $(\wt X_n)_{n\in\N_0}$ defined by
\begin{align*}
  \wt X_n & =  \wt X_{n-1} + \frac{\sigma^2}{2} \widehat \nabla^N  \log\pi_y(\wt X_{n-1}) + Z_n \\
   & = \left(1 -\frac{\sigma^2}{2 \sigma_p^2} \right)\wt X_{n-1} 
	  + \frac{\sigma^2}{2}\left( s(y) - \frac{1}{N} \sum_{i=1}^N s(Y_i)\right)
	  +Z_n.
 \end{align*}
The sequence $(\wt X_n)_{n\in \N_0}$ 
is again a Markov chain with transition kernel
\begin{align*}
  P_{\sigma,N}(\theta,A) 
& = \int_\R \sum_{(y'_1,\dots,y'_N)\in \mathcal{Y}^{MN}} 
 \mathbf{1}_A\left( \left(1-\frac{\sigma^2}{2\sigma_p^2}\right)\theta 
 + \frac{\sigma^2}{2}\left( s(y)-\frac{1}{N}\sum_{i=1}^N s(y_i') \right) + z \right)\\
&\qquad \qquad \qquad \qquad \qquad \qquad \qquad
\times \Pi_{i=1}^N \ell(\theta \,|\, y_i')\,
 \mathcal{N}(0,\sigma^2)(\dint z)
\end{align*}
for $\theta\in \R$ and $A\in \mathcal{B}(\R)$. 
Let us state a transition of this noisy Langevin Markov chain 
according to $P_{\sigma,N}$ in algorithmic form.
\begin{alg}
 A single transition from $\wt X_n$ to $\wt X_{n+1}$ works as follows:
 \begin{enumerate}
  \item \label{it: draw_Y}
  Draw an i.i.d. sequence $(Y_i)_{1\leq i\leq N}$ with $Y_i \sim \ell(\cdot \,|\, \wt X_n)$,
	  call the result $(y'_1,\dots,y'_N)$;
  \item Calculate 
   \[
    \widehat \nabla^{N} \log \pi_y(\wt X_n) 
      := s(y) -  \frac{1}{N} \sum_{i=1}^N s(y'_i) - \frac{\wt X_n}{\sigma_p^2};
   \]
  \item Draw $Z_n \sim \mathcal{N}(0,\sigma^2)$,
  independent from step \ref{it: draw_Y}., call the result $z_n$. Set
  \[
   \wt X_{n+1} = \wt X_n + \frac{\sigma^2}{2}   
   \widehat \nabla^{N} \log \pi_y(\wt X_n) + z_n.
  \]
 \end{enumerate}
\end{alg}
From \cite[Lemma~3]{AFEB14} 
 and by applying 
arguments of \cite{RoTw96b}, we obtain the following 
facts about the noisy Langevin algorithm.
\begin{prop} \label{prop: Langevin}
Let $\norm{s}_\infty = \sup_{z\in \mathcal{Y}^M} \abs{s(z)}$ be finite with $\norm{s}_\infty>0$, 
let $V\colon \R \to [1,\infty)$ be given by $V(\theta)=1+\abs{\theta}$ and assume that
$\sigma^2 < 4 \sigma^2_p$. Then
\begin{enumerate}
 \item \label{en: 3.}
 the function $V$ is a Lyapunov function for $P_\sigma$ and $P_{\sigma,N}$. 
 We have
 \begin{align} \label{eq: drift_Langevin_alg}
   P_\sigma V(\theta)  \leq \delta V(\theta) + L \mathbf{1}_I(\theta), \qquad
   P_{\sigma,N} V(\theta)  \leq \delta V(\theta) + L\mathbf{1}_I(\theta)
 \end{align}
 with $\delta = 1-\frac{\sigma^2}{4 \sigma_p^2}$, 
 $L=\sigma
 + \sigma^2
\norm{s}_\infty + \frac{\sigma^2}{2\sigma_p^2}$
 and the interval
 \[
  I=\left\{\theta\in \R \left| \abs{\theta} \leq 1+4\sigma_p^2 \norm{s}_\infty + \frac{4\sigma_p^2}{\sigma}\right. \right\}.
 \]
 \item 
 there are distributions $\pi_\sigma$ and $\pi_{\sigma,N}$ on $(\R,\mathcal{B}(\R))$
 which are stationary with respect to $P_\sigma$ and $P_{\sigma,N}$, respectively.
 \item \label{en: 1.}
 the transition kernels $P_\sigma$ and $P_{\sigma,N}$ are $V$-uniformly ergodic.
 \item \label{en: 4}
for $N > 4 \max\left\{ \|s\|_\infty^2 \sigma^4, 
\|s\|_\infty^{-3} \sigma^{-6} \right\}$ we have 
\begin{equation}\label{NLbound}
\sup_{\theta\in \R} \norm{P_{\sigma}(\theta,\cdot) - P_{\sigma,N}(\theta,\cdot)}_{{\rm tv}}
 \leq 
 6 \max \left\{
  \| s\|_\infty \sigma^2,
 \| s\|_\infty^{-2} \sigma^{-4}\right\}\, \frac{\log(N)}{N}.
\end{equation}
\end{enumerate}
\end{prop}
\begin{proof}
  We use the same arguments as in \cite[Section~3.1]{RoTw96b}.
  One can easily see that the Markov chains $(X_n)_{n\in \N_0}$ and $(\wt X_n)_{n\in \N_0}$ are irreducible with respect to 
  the Lebesgue measure and weak Feller.
  Thus, all compact sets are petite, see \cite[Proposition~6.2.8]{MeTw09}.
  Hence, for the existence of stationary distributions, say $\pi_\sigma$ and $\pi_{\sigma,N}$,
  \cite[Theorem~12.3.3]{MeTw09}, 
  as well as 
  for the $V$-uniform ergodicity \cite[Theorem~16.0.1]{MeTw09} 
  it is enough to show that $V$ satisfies \eqref{eq: drift_Langevin_alg}.
With $Z \sim \mathcal{N}(0,\sigma^2)$, we have 
\begin{align*}
   P_\sigma V(\theta) 
& \leq \left(1-\frac{\sigma^2}{2 \sigma_p^2}  \right)V(\theta) + \frac{\sigma^2}{2\sigma^2_p}
    + \frac{\sigma^2}{2}\abs{s(y) - \mathbb{E}_{\ell(\cdot\mid \theta)}s(Y)} 
    + \mathbb{E}\abs{Z}\\
  & \leq  \left(1-\frac{\sigma^2}{2 \sigma_p^2}  \right) V(\theta) 
    + \frac{\sigma^2}{2\sigma_p^2}
    + \sigma^2 \norm{s}_\infty 
    + \sigma \\
 & \leq \left(1-\frac{\sigma^2}{2 \sigma_p^2}  \right) V(\theta) 
   + \max\left\{\frac{\sigma^2}{4 \sigma_p^2} V(\theta), \frac{\sigma^2}{2\sigma_p^2}
   +\sigma^2 \norm{s}_\infty+\sigma \right\} \\
 & \leq \left(1-\frac{\sigma^2}{4 \sigma_p^2}  \right) V(\theta) 
 +  \left( \frac{\sigma^2}{2\sigma_p^2}+\sigma^2\norm{s}_\infty+\sigma\right )
 \cdot \mathbf{1}_{I}(\theta).
 \end{align*}
By the fact that 
\[
\mathbb{E}\left[\abs{s(y) - \frac{1}{N} \sum_{i=1}^N s(Y_i)} \mid \wt X_n=\theta\right] \leq 2 \norm{s}_\infty
\]
we obtain with the same arguments that
\[
 P_{\sigma,N} V(\theta) \leq 
 \delta V(\theta) + L\cdot \mathbf{1}_I(\theta).
\]
Thus, the assertions from \ref{en: 3.}. to \ref{en: 1.}. are proven. 
The statement of
\ref{en: 4}. is a consequence of \cite[Lemma~3]{AFEB14}. 
There it is shown that for $N >4 \|s\|_\infty^2 \sigma^4$ it holds that 
\[
\sup_{\theta\in \R} \norm{P_{\sigma}(\theta,\cdot) - P_{\sigma,N}(\theta,\cdot)}_{\text{tv}} \leq \exp\left(\frac{\log(N)}{4 N \|s\|_\infty^2 \sigma^4} \right)-1 + \frac{4\sqrt{\pi} \|s\|_\infty \sigma^2}{N}.
\]
By using $\exp(\theta)-1 \leq \theta \exp(\theta)$ and $N>4$ we further estimate the 
right-hand side by
\[
\left(
\frac{K_{N,s,\sigma}}{4 \|s\|_\infty^2 \sigma^4} 
+ \frac{4\sqrt{\pi} \|s\|_\infty \sigma^2}{\log(5)} 
\right) \cdot\frac{\log{(N)}}{N}
\quad\text{ with }\quad 
K_{N,s,\sigma} = \exp\left(\frac{\log(N)}{4 N \|s\|_\infty^2 \sigma^4} \right).
\]
Since $ \log(N) \cdot N^{-1/3} < 2$, 
we have the bound  $K_{N,s,\sigma} \leq \exp(1)$ provided 
that $4 N^{2/3} \|s\|_\infty^2 \sigma^4 \geq 2$ 
which follows from $N \geq  \|s\|_\infty^{-3} \sigma^{-6}$. 
The assertion of \eqref{NLbound} follows now by a simple calculation.
\end{proof}

By using the facts collected in the previous proposition, 
we can apply the perturbation bound of Theorem~\ref{thm geom3} 
and obtain a quantitative perturbation bound for the noisy Langevin algorithm. 

\begin{cor}
 Let $p_0$ be a probability measure on $(\R,\mathcal{B}(\R))$ 
 and set $p_n=p_0 P_\sigma^n$ as well as $\wt p_{n,N} = p_0 P_{\sigma,N}^n$. Suppose that $\sigma^2 < 4 \sigma_p^2$.
 Then, there are numbers $\rho\in[0,1)$ and $C\in(0,\infty)$, independent of $n,N$,
 determining 
 \begin{equation*}
   R :=\frac{18 \max\{ \norm{s}_\infty \sigma^2,\norm{s}_\infty^{-2} \sigma^{-4} \}}{1-\rho} 
   \cdot \left( 2+\max\left\{ \mathbb{E}_{p_0}\abs{X},4 \sigma_p^2(\norm{s}_\infty+\sigma^{-1}) \right\} \right)
 \end{equation*}
 with $\mathbb{E}_{p_0}\abs{X} = \int_\R \abs{\theta} \,\dint p_0(\theta)$,
 so that for $N>90\max\{ \norm{s}_\infty^2 \sigma^4,\norm{s}_\infty^{-3} \sigma^{-6} \}$
 we have
 \begin{equation*}
     \max\left\{ \norm{p_n - \wt p_{n,N}}_{\text{\rm tv}},  
  \norm{ \pi_\sigma - \pi_{\sigma,N}}_{\text{\rm tv}} \right\}
  \leq R \cdot 
  \left(2C \left(\sigma + \sigma^2 \norm{s}_\infty +3  \right) \right)^{2/\log(N)} 
  \frac{\log(N)^2}{N}.
 \end{equation*}
\end{cor}
\begin{proof} 
We have by Proposition~\ref{prop: Langevin} that $P_\sigma$ is $V$-uniformly ergodic with $V(\theta)=1+\abs{\theta}$, i.e., 
there are numbers $\rho\in [0,1)$ and $C\in (0,\infty)$ such that 
 \[
  \sup_{\theta\in \R} \frac{\norm{P_{\sigma}^n(\theta,\cdot)-\pi_\sigma}_{V}}{V(\theta)} \leq C \rho^n .
 \]
Now, by combining Theorem~\ref{thm geom3} and Remark~\ref{rem: thm geom4} with the results from Proposition~\ref{prop: Langevin}
we obtain the result.
\end{proof}

\begin{rem}
We want to point out that the assumptions imposed are the same 
as in \cite[Theorem~3.2]{AFEB14}, but 
instead of the asymptotic result 
we provide an explicit estimate.
The numbers $\rho \in [0,1)$
and $C\in (0,\infty)$ are not stated in terms of the model parameters. 
In principle, these values can be 
derived from the drift condition \eqref{eq: drift_Langevin_alg} 
through \cite[Theorem~1.1]{Ba05}.
\end{rem}

\section*{Acknowledgements}
We thank Alexander Mitrophanov 
and the referees for their valuable comments which helped to improve the paper. D.R. was
supported by the DFG Research Training Group 2088.

\bibliographystyle{imsart-number}

\end{document}